\documentclass[12pt]{article}
\usepackage{amsmath}
\usepackage{graphics}
\usepackage{xcolor}
\usepackage{amsfonts}
\usepackage{times}
\usepackage{bm}
\usepackage{natbib}
\usepackage{xr}
\usepackage{booktabs}
\usepackage{amsthm}
\usepackage{tikz}
\usepackage[flushleft]{threeparttable}
\usepackage{longtable}
\usetikzlibrary{shapes,decorations,arrows,calc,arrows.meta,fit,positioning}
\tikzset{
    -Latex,auto,node distance =1 cm and 1 cm,semithick,
    state/.style ={ellipse, draw, minimum width = 0.7 cm, fill = yellow!25},
    point/.style = {circle, draw, inner sep=0.04cm,fill,node contents={}},
    bidirected/.style={Latex-Latex,dashed},
    el/.style = {inner sep=2pt, align=left, sloped}
}
\usepackage[ruled,vlined]{algorithm2e}

\makeatletter
\renewcommand{\algocf@captiontext}[2]{#1\algocf@typo. \AlCapFnt{}#2} 
\def\@algocf@capt@plain{top}
\renewcommand{\algocf@makecaption}[2]{%
  \addtolength{\hsize}{\algomargin}%
  \sbox\@tempboxa{\algocf@captiontext{#1}{#2}}%
  \ifdim\wd\@tempboxa >\hsize
    \hskip .5\algomargin%
    \parbox[t]{\hsize}{\algocf@captiontext{#1}{#2}}
  \else%
    \global\@minipagefalse%
    \hbox to\hsize{\box\@tempboxa}
  \fi%
  \addtolength{\hsize}{-\algomargin}%
}
\makeatother

\def\sq{{ \mathrm{\scriptscriptstyle (SQ)} }}
\def\me{{ \mathrm{\scriptscriptstyle (ME)} }}
\def\logit{{ \mathrm{\scriptscriptstyle (logit)} }}
\def\T{{ \mathrm{\scriptscriptstyle T} }}
\def\P{{ \mathbb{\scriptscriptstyle P} }}
\def\dr{{ \mathrm{bc} }}
\def\pr{{ \mathrm{pr} }}
\def\var{{ \mathrm{var} }}

\def\rmsc{{ \mathrm{\scriptscriptstyle SC} }}

\def\bbU{{\mathbb{U}}}
\def\bx{{x}}
\def\bu{{u}}
\def\E{{E}}
\def\bbS{{\mathbb{S}}}
\newcommand*{\indep}{%
{\mbox{$\perp\!\!\!\perp$}}
}

\makeatletter
\newcommand*{\addFileDependency}[1]{
  \typeout{(#1)}
  \@addtofilelist{#1}
  \IfFileExists{#1}{}{\typeout{No file #1.}}
}
\makeatother

\newtheorem{lemma}{Lemma}\newtheorem{theorem}{Theorem}
\newtheorem{assumption}{Assumption}\newtheorem{remark}{Remark}
\newtheorem{corollary}{Corollary}
\newtheorem{example}{Example}

\newtheorem{proposition}{Proposition}
\begin{document}

\def\spacingset#1{\renewcommand{\baselinestretch}%
{#1}\small\normalsize} \spacingset{1}

\newcommand{\blind}{1}
\if1\blind
{
\date{}
  \title{\bf Soft calibration for selection bias problems under mixed-effects models}
  \author{Chenyin Gao   \thanks{ Department of Statistics, North Carolina State University},\hspace{.2cm}Shu Yang \thanks{
   Department of Statistics, North Carolina State University, North Carolina
27695, U.S.A. Email: syang24@ncsu.edu},\hspace{.2cm}and Jae Kwang Kim\thanks{
  Department of Statistics, Iowa State University}
    }
  \maketitle
} \fi

\if0\blind
{
  \bigskip
  \bigskip
  \bigskip
  \begin{center}
    {\LARGE\bf  Soft calibration for selection bias problems under mixed-effects models}
\end{center}
  \medskip
} \fi

\bigskip

\begin{abstract}
Calibration weighting has been widely used to correct selection biases in non-probability sampling, missing data, and causal inference. The main idea is to calibrate the biased sample to the benchmark by adjusting the subject weights. However, hard calibration can produce enormous weights when an exact calibration is enforced on a large set of extraneous covariates. This article proposes a soft calibration scheme, in which the outcome and the selection indicator follow mixed-effects models. The scheme imposes an exact calibration on the fixed effects and an approximate calibration on the random effects. On the one hand, our soft calibration has an intrinsic connection with best linear unbiased prediction, which results in a more efficient estimation compared to hard calibration. On the other hand, soft calibration weighting estimation can be envisioned as penalized propensity score weight estimation, with the penalty term motivated by the mixed-effects structure. The asymptotic distribution and a valid variance estimator are derived for soft calibration. We demonstrate the superiority of the proposed estimator over other competitors in simulation studies and a real-data application. 
\end{abstract}
\noindent%
{\it Keywords:}   Inverse propensity score weighting; Latent ignorability; Penalized optimization; Restricted maximum likelihood estimation.
\vfill

\newpage
\spacingset{1.5} 

\section{Introduction}\label{sec:introduction}

Calibration weighting, or benchmark weighting, is popular in survey sampling, where probability sampling weights are adjusted to match the known population totals of the auxiliary variables for a possible efficiency gain \citep{deville1992calibration}. The idea of calibration is related to the generalized regression estimator, a model-assisted estimator in survey sampling \citep{cassel1976some,sarndal1992}, which has later been extended to the functional model-assisted estimator \citep{cardot2011horvitz}, optimal model calibration \citep{wu2001model}, calibration weighting using instrumental variables \citep{estevao2000functional}, empirical likelihood calibration \citep{wu2006pseudo}, and multi-source data calibration \citep{yang2019combining}.

In addition to gaining precision, calibration weighting has been widely used to correct selection bias in various contexts, including finite-population inferences using non-probability samples, missing data, and causal inference. \citet{skinner1999calibration}, \citet{lundstrom1999calibration}, \citet{deville2000generalized}, \citet{kott2006using} and \citet{Lee2009} employed calibration weighting to adjust for selection bias in non-probability samples by enforcing covariate similarity between the non-probability sample and a probability sample; see \cite{yang2020statistical} for a comprehensive review. For missing-at-random data, inverse propensity score weighting creates a weighted sample that resembles the complete version of the original sample. Instead of directly inverting the propensity score, calibration weighting imposes conditions to emulate complete data and gains robustness against model misspecification (\citealp{han2013estimation,chen2017multiply,dong2020integrative,lee2022doubly}). Similarly, for causal inference under the ignorability of treatment assignment, the purpose of calibration weighting is to achieve the covariate balance between treatment groups, thus mitigating confounding biases
(\citealp{hainmueller2012entropy, anastasiade2017decomposition}). For example, the covariate balance propensity score introduced by \citet{imai2014} uses a balancing measure as an objective function to estimate the propensity score.

Most existing works aim to calibrate all available auxiliary variables to known finite-population totals, a process known as hard calibration. However, hard calibration may not be necessary when there are many covariates, especially if some covariates are not predictive of the outcome. Over-calibration, or improper application of calibration weighting on too many variables, can lead to variance inflations \citep{kang2007demystifying}. To address this problem, subsequent research has sought to use penalization (\citealp{guggemos2010penalized,athey2018approximate,ning2020robust}) or regularization \citep{zubizarreta2015,wong2018,wang2022estimation} to ease the calibration constraints on a subset of covariates, which we refer to as regularized calibration. \citet{chattopadhyay2020balancing} proposed minimal dispersion approximately balancing weights by optimizing some user-specified function. Other attempts have been made to reduce the range of calibration weights directly by trimming, smoothing, or stabilizing \citep{lazzeroni1998random,yang2018asymptotic}. 
Many of these methods adopt mixed-effects modeling, which is particularly useful in small area estimation \citep{torabi2008small}, longitudinal data inference \citep{verbeke1997linear,weiss2005modeling}, handling clustered data with cluster-specific nonignorable missingness \citep{kim2016b}, and causal inference with unmeasured cluster-level confounders \citep{yang2018propensity}. 

In this article, we focus on the settings with the shared parameter/random-effects models of the outcome and the selection indicator \citep{follmann95}. The sample inclusion indicator in survey sampling, the response indicator in the missing data context, and the treatment assignment in causal inference are all examples of the selection indicator. As a result, our framework applies to a wide range of problems. The selection indicator in the shared parameter models is latently ignorable in the sense that the selection indicator and outcome are conditionally independent given the observed covariates and the unobserved random effects, entailing nonignorable selection. Under the linear mixed-effects model, we propose a soft calibration algorithm that enforces an exact calibration on fixed effects, see (\ref{eq:calib_fix}), and an approximate calibration on the random effects, see (\ref{eq:calib-random}). Our soft calibration exploits the correlation structure of random effects to construct the regularized constraints, which is different from typical regularized calibration methods that leverage sparsity or smoothness conditions \citep{tan2020model,ning2020robust}. The soft calibration constraints are seemingly intricate but arise naturally from two paths towards constructing the best linear unbiased predictor $\widehat{\theta}_{\rm blup}$, a minimization problem in (\ref{eq:soft_sq}) and a prediction approach in (\ref{eq:estimating-eqs}). Thus, the produced estimator has an intrinsic connection to $\widehat{\theta}_{\rm blup}$ and can be more efficient than the hard-calibration estimator, especially when random effects weakly affect the outcome. Furthermore, the dual problem (\ref{eq:est_alpha_beta}) of soft calibration also establishes a link between soft calibration and penalized propensity score weight estimation, leading to a ridge-type regression \citep{guggemos2010penalized}. 

The calibration weights are well-known to be obtained by optimizing the user-specified loss function, which is related to the modeling of the propensity scores. Because the constrained optimization formulation (\ref{eq:constrained}) separates the loss function from the calibration conditions, we can impose relaxed calibration conditions while forcefully bounding the range of weights by changing the loss function. Next, we can show that the soft-calibration estimator is consistent if either the outcome follows a linear mixed-effects model or the propensity score model is correctly specified. The asymptotic distribution and a valid variance estimator for the soft-calibration estimators are then established. Furthermore, augmentations with flexible outcome modeling can be used in conjunction with soft calibration to correct the remaining bias, if any. Finally, a data-adaptive approach aided by cross fitting is proposed to select the optimal tuning parameter that minimizes the finite-sample mean squared error. Proofs of all results are provided in the Supplementary Material.

\section{Basic setup\label{sec:Basic-setup} }

\subsection{Notations, ignorability, and hard calibration\label{subsec:notations}}

To fix ideas, we consider estimating the population mean of a study variable based on a non-probability sample and extend it to clustered missing data analysis in $\S$\ref{subsec:Cluster-specific-nonignorable-mi}. Suppose that we have a finite population $\mathcal{F}_{N}=\{(\bx_{i},y_{i}):i\in\bbU\}$ with population size $N$ and index set $\bbU=\{1,\ldots,N\}$, independently and identically following a super-population model $\zeta$. We assume that $\bx_{i}$ is available in the finite population, but the study variable $y_{i}$ is observed only in the sample. Let $\mathbb{S}\subset\mathbb{U}$ be the index set of the sample of size $n$. Define the selection indicator $\delta_{i}$ as $\delta_{i}=1$ if $i\in\bbS$ and $0$ otherwise. The propensity score for unit $i$ being selected in the sample is $\pi_{i}=\pr(\delta_{i}=1\mid\bx_{i})$, which is unknown for the non-probability sample. For ease of presentation, we summarize all notations in Table \ref{tab:notations} for reference.
\begin{table}[!ht]
\caption{ Summary of the notations\label{tab:notations}}
\vspace{0.15cm}
\resizebox{\textwidth}{!}{
      \begin{tabular}{ll}
      \toprule
          Notation & Definition \\
        \midrule
          $y_i, x_i, x_{1i}, x_{2i}$ & Individuals of study variable and covariate for unit $i$, $x_i=(x_{1i}^{\T}, x_{2i}^{\T})^{\T}$\\
          $Y_{\bbU},Y_{\bbS}$& Vectors of study variable, $Y_{\bbU}=(y_1,\cdots,y_N)^{\T}$, $Y_{\bbS} = 
          \{y_i:i\in\bbS\}$\\
          $X_{\bbU}, X_{1,\bbU}, X_{2,\bbU}$ & Matrices of covariate for finite population $\bbU$, $X_{\bbU}=(X_{1,\bbU}, X_{2,\bbU})\in \mathbb{R}^{N\times (p+q)}$\\
          $X_{\bbS}, X_{1,\bbS}, X_{2,\bbS}$ & Matrices of covariate for selected sample $\bbS$, $X_{\bbS}=(X_{1,\bbS}, X_{2,\bbS})\in\mathbb{R}^{n\times (p+q)}$\\
          $\E_{\delta}(\cdot), \E_{\zeta}(\cdot), \E(\cdot)$&
          Expectations with respect to the selection $\delta$, the model $\zeta$, and both\\
          $\var_{\delta}(\cdot), \var_{\zeta}(\cdot), \var(\cdot)$&
          Variances with respect to the selection $\delta$, the model $\zeta$, and both\\
          $o(\cdot)$&
          $a_n=o(b_n)$ implies $a_n/b_n\rightarrow 0$ when $n\rightarrow \infty$\\
          $O(\cdot)$&
          $a_n=O(b_n)$ implies $a_n/b_n\rightarrow C_0$ when $n\rightarrow \infty$ for some constant $C_0$\\
          $o_{\P}(\cdot), O_{\P}(\cdot)$&
          Small and big order terms with respect to both the selection $\delta$ and model $\zeta$\\
          \bottomrule
      \end{tabular}}
  \end{table}
  
The goal is to estimate $\theta_N=N^{-1}\sum_{i\in\bbU}y_{i}$, and we consider a weighted estimator given by
\begin{equation}
\widehat{\theta}_{w}=\frac{1}{N}\sum_{i\in\bbS}w_{i}y_{i}.\label{eq:weighting estimator}
\end{equation}
If $y_{i}$ follows the linear regression model $y_{i}=\bx_{i}^{\T}\beta+e_{i}$ with $\E_{\zeta}(e_{i}\mid\bx_{i})=0$ and $\var_{\zeta}(e_i\mid \bx_{i})=\sigma_e^2$, we may impose the following condition on the weights:
\begin{equation}
\sum_{i\in\bbS}w_{i}\bx_{i}=\sum_{i\in\bbU}\bx_{i},\label{calib}
\end{equation}
which is a sufficient condition for model calibration \citep{wu2001model} in the sense that $\sum_{i\in\bbS}w_{i}\widehat{y}_{i}=\sum_{i\in\bbU}\widehat{y}_{i}$, where $\widehat{y}_i$ is a prediction based on the linear model. If the sampling mechanism is ignorable with $\delta_{i}\indep y_{i}\mid\bx_{i}$, condition (\ref{calib}) is sufficient for the unbiasedness of $\widehat{\theta}_{w}$. To find the optimal calibration estimator that minimizes the mean squared error of $\widehat{\theta}_{w}$ while satisfying (\ref{calib}) under the linear regression model, it suffices to minimize
$$
\E_{\zeta}\{(\widehat{\theta}_w-\theta_N)^2\mid X_{\bbU},\bbS\}=\frac{1}{N^2}\var_{\zeta}\left\{ \sum_{i\in\bbU}(\delta_{i}w_{i}-1)e_{i}\mid X_{\bbU},\bbS \right\} =\frac{\sigma_{e}^{2}}{N^2}\sum_{i\in\bbS}(w_{i}-1)^{2}+{\rm const.,}
$$
where ${\rm const.}$ represents a constant that does not depend on
$w=\{w_{i}:i\in\bbS\}$. Thus, we can formulate the hard
calibration weighting problem as finding the minimizer of the square loss function $\sum_{i\in\bbS}(w_{i}-1)^{2}$
subject to condition (\ref{calib}).

\subsection{Mixed-effects models and latent ignorability \label{subsec:Latent-ignorability-and}}

We now partition $\bx_{i}$ into two vectors $\bx_{1i}$ (including an intercept) and $\bx_{2i}$
with $\mbox{dim}(\bx_{1i})=p$ and $\mbox{dim}(\bx_{2i})=q$, related
to fixed effects and random effects, respectively. This setup
is particularly relevant in small area estimation, where $\bx_{1i}$
is a low-dimensional vector of feature variables and $\bx_{2i}$ is
a possibly high-dimensional vector of small area indicators. 

In these settings, selection ignorability can be restrictive because it excludes area-specific effects that affect both $y_i$ and $\delta_i$. To overcome this issue, we consider a linear
mixed-effects super-population model:
\begin{equation}
y_i = x_{1i}^{\T}\beta + 
x_{2i}^{\T} u + e_i, \quad\bu\sim N\left(0,D_{q}\sigma_{u}^{2}\right),\quad e_i\sim N(0,q_i^{-1}\sigma_{e}^{2}),\quad\bu\indep e_i \mid x_i,\label{mixed}
\end{equation}
where $\bu$ is a $q$-dimensional vector of random effects with a positive-definite covariance matrix $D_{q}$, $e_i$ is the heteroscedastic random error with known $q_i^{-1}$, and $\sigma_{e}^{2}$ and $\sigma_{u}^{2}$ characterize the variances of individual errors and random effects, respectively. Typically, we consider $q_i=1$ for $i\in \bbS$ but unequal $q_i$'s are also desired in some situations; see Remark 5 in \cite{devaud2019deville}. Next, we make the following assumptions for the sampling
mechanism.
\begin{assumption}[Latent ignorability]\label{assum:laten_ignore}
The sampling mechanism is ignorable given $(\bx_{i},\bu)$: $\delta_{i}\indep y_{i}\mid(\bx_{i},\bu)$ for all $i\in\bbU$.
\end{assumption}
\begin{assumption}[Positivity]\label{assum:positivity}
$0<\underline{d}<Nn^{-1}\pr(\delta_{i}=1\mid\bx_{i},\bu)<\overline{d}<1$
for all $\bx_{i}$ and $\bu$.
\end{assumption}
Assumption \ref{assum:laten_ignore} leads to shared parameter/random-effects
models of $\delta_{i}$ and $y_{i}$. In the missing data context
with clustered data, it is called
cluster-specific nonignorable missingness \citep{yuan2007model}.
In the context of causal inference, it
is called cluster-specific nonignorable treatment assignment \citep{yang2018propensity}.
Assumption \ref{assum:laten_ignore} relaxes the ignorability assumption by allowing unobserved random effects to affect both $y_i$ and $\delta_i$. Assumption \ref{assum:positivity}
implies that the sample support $\{x_{i}: i\in\bbS\}$
coincides with the support of $x_{i}$ in the population.

\subsection{Soft calibration for the best linear unbiased predictor}\label{subsec:soft-blup}

Under model (\ref{mixed}) and Assumptions \ref{assum:laten_ignore}-\ref{assum:positivity}, we wish to develop the optimal calibration estimator $\widehat{\theta}_{w}$ by minimizing the mean squared error. Following \citet{hirshberg2019minimax}'s minimax imbalance strategy, we minimize
\begin{align}
    &\sup_{\beta \in \mathcal{M}} 
    E_{\zeta}\{(\widehat{\theta}_w-\theta_N)^2\mid X_{\bbU},\bbS\} =
  \sup_{\beta\in\mathcal{M}}
  \frac{1}{N^2}
  (w^{\T} X_{1,\bbS}-1^{\T}_N X_{1,\bbU})\beta\beta^{\T}
    (w^{\T} X_{1,\bbS}-1^{\T}_N X_{1,\bbU})^{\T}\nonumber\\
    &
    +\frac{\sigma_e^2}{N^2}\left\{\sum_{i \in \bbS}q_i^{-1}\left(w_i-1\right)^2+\gamma^{-1}\left(w^{\T} X_{2, \bbS}-1_N^{\T} X_{2,\bbU}\right) D_q\left(w^{\T} X_{2, \bbS}-1_N X_{2,\bbU}\right)^{\T}\right\}
    \label{eq:soft_sq}
\end{align}
with respect to $w$, where $\mathcal{M}$ is a convex subset of $\mathbb{R}^p$ that contains the true $\beta$. Since $\mathcal{M}$ may be unbounded without prior knowledge, the minimax problem results in an exact calibration condition $w^{\T} X_{1, \bbS}=1^{\T}_N X_{1,\bbU}$ to diminish the first term of the above equation. The remaining objective function (\ref{eq:soft_sq}) leads to a generalized ridge regression problem \citep{bardsley1984multipurpose} augmented with a data-dependent penalty, where $\gamma^{-1}=\sigma_{u}^{2}/\sigma_{e}^{2}$ determines the level of calibration for $X_{2,\bbS}$: if $\gamma$ is close to zero, the
calibration for $X_{2,\bbS}$ is nearly exact; and if $\gamma$ is large, the calibration for $X_{2,\bbS}$ is greatly relaxed.

In addition, the minimum of (\ref{eq:soft_sq}) should coincide with $\widehat{\theta}_{\rm blup}=N^{-1}\sum_{i\in \bbU}(x_{1i}^{\T}\widehat{\beta}+x_{2i}^{\T}\widehat{u})$,
where $(\widehat{\beta},\widehat{u})$ is the solution to the following score equations for the linear mixed-effects model: 
\begin{equation}
\begin{pmatrix}\sum_{i\in\bbS}q_i\bx_{1i}\bx_{1i}^{\T} & \sum_{i\in\bbS}q_i\bx_{1i}\bx_{2i}^{\T}\\
\sum_{i\in\bbS}q_i\bx_{2i}\bx_{1i}^{\T} & \sum_{i\in\bbS}q_i\bx_{2i}\bx_{2i}^{\T}+\gamma D_{q}^{-1}
\end{pmatrix}\begin{pmatrix}\beta\\
u
\end{pmatrix}=\begin{pmatrix}\sum_{i\in\bbS}q_i\bx_{1i}y_{i}\\
\sum_{i\in\bbS}q_i\bx_{2i}y_{i}
\end{pmatrix}.\label{eq:estimating-eqs}
\end{equation}
By rewriting $\widehat{\theta}_{\rm blup}$ as a weighted estimator $w^{\T}Y_{\bbS}$, the weights satisfy
\begin{align*}
    w^{\T}X_{\bbS}& = 
    1_{N}^{\T} X_{\bbU} 
    \left\{\sum_{i\in \bbS}q_i x_{i}x_{i}^{\T}+\gamma\text{diag}(0,D_{q}^{-1})\right\}^{-1} 
    \sum_{i\in\bbS} q_i x_{i}x_{i}^{\T}\\
    & = 
    1_{N}^{\T} X_{\bbU}
    \left[
    I_{p+q} - 
    \gamma\left\{\sum_{i\in \bbS}q_i x_{i}x_{i}^{\T}+\gamma\text{diag}(0,D_{q}^{-1})\right\}^{-1} 
    \text{diag}(0,D_{q}^{-1})
    \right],
\end{align*}
where the second equality is derived by repeatedly applying the Woodbury matrix identity. Therefore, minimizing (\ref{eq:soft_sq}) can be reformulated as a constrained optimization with exact calibration on $\bx_{1i}$ and approximate calibration on $\bx_{2i}$:
\begin{subequations}
\begin{align}
\min_{w}\quad &\sum_{i\in\bbU}\delta_i Q(w_i)=\sum_{i\in\bbS} q_i^{-1}(w_i-1)^2,\nonumber\\
 \text{s.t. } &\sum_{i\in\bbS}w_{i}x_{1i}=\sum_{i\in\bbU}x_{1i},\label{eq:calib_fix}\\
 & \sum_{i\in\bbS}w_{i}\bx_{2i}=\sum_{i\in\bbU}\bx_{2i}
 + \sum_{i\in\bbU}M_{\bbS}^{\T}\bx_{1i}+\sum_{i\in\bbU}R_{\bbS}^{\T}\bx_{2i},\label{eq:calib-random}
\end{align}
\label{eq:constrained}\end{subequations}
where $M_\bbS = -\gamma D_{12}D_{q}^{-1}$, $R_\bbS = -\gamma D_{22}D_{q}^{-1}$, and $\{\sum_{i\in \bbS}q_i x_{i}x_{i}^{\T}+\gamma\text{diag}(0,D_{q}^{-1})\}^{-1} = 
[D_{11},D_{12}\mid
D_{21},D_{22}]$. The solution is denoted by $\widehat{w}^{\sq}=\{\widehat{w}_i^{\sq}:i\in\bbS\}$, giving rise to $\widehat{\theta}_w^{\sq}=N^{-1}\sum_{i\in\bbS}$ $\widehat{w}_i^{\sq}y_i$, where the superscript \textsc{sq} reflects the use of the square loss. 

Proposition \ref{prop:BLUP-penalized} reveals the intrinsic connection between soft calibration based on square loss and $\widehat{\theta}_{\rm blup}$ under the mixed-effects model (\ref{mixed}).
\begin{proposition}\label{prop:BLUP-penalized}
Under Assumptions \ref{assum:laten_ignore} and \ref{assum:positivity} and the model (\ref{mixed}), we have $\widehat{\theta}_{w}^{\sq}=\widehat{\theta}_{{\rm blup}}$ for fixed $\gamma=\sigma_e^2/\sigma_u^2$.
\end{proposition}

Through the lens of $\widehat{\theta}_{\rm blup}$ derived from (\ref{eq:soft_sq}) or (\ref{eq:estimating-eqs}), the soft-calibration estimator is optimal under model (\ref{mixed}) and consistent under any sampling design that satisfies the latent ignorability by Proposition \ref{prop:BLUP-penalized}. 

\subsection{Soft calibration for penalized propensity score weight estimation \label{subsec:soft_cal_loss}}

In the proof of Proposition \ref{prop:BLUP-penalized}, we show that the square loss function is equivalent to assuming a linear regression model for the calibration weight. However, it is possible to obtain negative values that may not be acceptable to practitioners. One advantage of casting the soft-calibration estimator as a solution to the constrained optimization
problem (\ref{eq:constrained}) is that it directly leads to a mixed-effects model for the calibration weight through the link function $w(\cdot)$, which allows flexible estimation by adopting other loss functions $Q(\cdot)$. In particular, we consider the \textit{dual problem} of (\ref{eq:constrained}) for optimization purposes, which is to minimize a penalized convex function:
\begin{subequations}
\label{eq:est_alpha_beta}
\begin{align}
G(c) &= - 
\sum_{i\in\bbU}\delta_i Q\{w(c^{\T}x_i)\} + 
\left\{\sum_{i\in\bbS}w(c^{\T}x_i)x_{i}^{\T}-
(1_N^{\T}X_{1,\bbU}, 1_N^{\T}X_{2,\bbU}+NT_r)\right\}c\label{eq:soft-penalized}\\
&=\sum_{i\in\bbU}\delta_i g(c^{\T}x_i) - 
(1_N^{\T}X_{1,\bbU})c_1-
(1_N^{\T}X_{2,\bbU} + NT_r)c_2, \label{eq:soft-dual}
\end{align}
\end{subequations}
where $g(\cdot)$ is the convex conjugate function of $Q(\cdot)$, $T_r=N^{-1}\sum_{i\in\bbU}(\bx_{1i}^{\T}M_{\bbS}+\bx_{2i}^{\T}R_{\bbS})$ is the adjustment for soft calibration, and $c=(c_1^\T, c_2^\T)^\T$ is a vector of Lagrange multipliers with $c_2 = D_{\delta}\bu$ for a suitable invertible matrix $D_{\delta}$, featuring a shared random-effects model with the outcome \citep{gao2004shared}. Table \ref{tab:obj_propensity} provides some examples of loss functions $Q(\cdot)$ and their associated $g(\cdot)$ and $w(\cdot)$. These loss functions belong to a general class of empirical minimum discrepancy measures \citep{read2012goodness}, which can be considered as measuring the aggregate distance between the weights $w$ and a $n$-vector of uniform weights $1_{n}$.

\begin{table}[htbp]
\caption{Correspondence of loss functions $Q(w_i)$, the convex conjugate functions $g(z_i)$ and the weight models $w(z_i)$ when weights are adjusted to satisfy the calibration constraints for the first moments of $\bx_{i}$ 
\label{tab:obj_propensity}}
\vspace{0.15cm}
\resizebox{\textwidth}{!}{%
\begin{tabular}{llllll}
\toprule
  & $Q(w_i)$& $g(z_i)$ & $w(z_i)$  \\
  \midrule
Squared loss & $q_i^{-1}(w_{i}-1)^{2}/2$  & $z_i + q_iz_i^2/2$& $1+q_iz_i$\\
Entropy divergence  & $q_i^{-1}\{w_{i}\log(w_{i})-w_{i}+1\}$ & $q_i^{-1}\{\exp(q_i z_i)-1\}$& $\exp(q_iz_{i})$ \\
Empirical Likelihood & $q_i^{-1}\{-\log(w_{i})-1+w_{i}\}$ & 
$-q_i^{-1}\log(1-q_i z_i)$ & $(1-q_iz_{i})^{-1}$ \\
Maximum entropy & $q_i^{-1}(w_{i}-1)\{\log(w_{i}-1)-1\}$ & 
$z_i + q_i^{-1}\exp(q_i z_i)$& $1+\exp(q_iz_{i})$ \\
\bottomrule
\end{tabular}}
\end{table}

\begin{proposition}\label{prop:dual-soft}
    If $\widehat{c}$ is the minimizer of (\ref{eq:soft-dual}), the calibration weights $w(\widehat{c}^{\T}x_{i})$ attain the soft calibration conditions (\ref{eq:calib_fix}) and (\ref{eq:calib-random}). 
\end{proposition}
Proposition \ref{prop:dual-soft} is justified since (\ref{eq:soft-dual}) gives a dual optimization for solving the constrained optimization in (\ref{eq:constrained}). Furthermore, the penalized estimation in (\ref{eq:soft-penalized}) is closely related to the $L_2$ penalized propensity score weight estimator, which is, however, not optimal as its penalty term does not account for the correlation structure of the mixed effects; see $\S$\ref{sec:sim_res_L_2} of the Supplementary Material for numerical details. In view of the Lagrangian function (\ref{eq:soft-penalized}), the soft-calibration estimator enforces an exact calibration on $x_{1i}$ while penalizing a large discrepancy of imbalances between $\sum_{i\in \bbS} w_i \bx_{2i}$ and $\sum_{i\in\bbU} \bx_{2i}$, thus avoiding posing overly stringent constraints.

\begin{remark}\label{rmk:unique}
Let $\mathbb{A}=\{w:w^{\T}X_{\bbS}=1_N^{\T}X_{\bbU}+(0_p^{\T},NT_r)\}$ be a set of solutions to the soft calibration conditions. Assume that $Q(w)$ is strictly convex and smooth, defined in $\mathbb{W}$ that includes $1$. Assume that $\mathbb{W}$ is either a compact set or an open set with $\lim_{w\rightarrow \partial \mathbb{W}}|Q(w)|=\infty$, where $\partial \mathbb{W}$ denotes the boundary of the set $\mathbb{W}$, (\ref{eq:est_alpha_beta}) has a unique optimum with probability $1$ when $\mathbb{A}\cap \mathbb{W}\not=\emptyset$.
\end{remark}

In finite samples, a unique optimum of (\ref{eq:est_alpha_beta}) may not exist due to conflicting conditions imposed for calibration. For example, calibration weights are restricted to an overly bounded support $\mathbb{W}$ to reduce the impact of outliers; see $\S$\ref{sec:sim_res_bound}, which might render $\mathbb{A}\cap \mathbb{W}$ empty. One remedy for this issue is to adopt a Moore-Penrose generalized inverse \citep{devaud2019deville} for the Newton-type method to achieve a solution even when $\mathbb{A}\cap \mathbb{W}=\emptyset$.

\section{Main theory}\label{sec:Main-theory}

\subsection{Bias correction and asymptotic properties}\label{sec:asymptotic}

In this section, we establish the asymptotic properties of $\widehat{\theta}_{w}$ under the general loss function $Q(w)$ and adopt the joint randomization framework for inference, which considers both the super-population mixed-effects model $\zeta$ and the sampling mechanism $\delta$ \citep{isaki1982}. Before delving into the technical details, we assume the following regularity conditions.

\begin{assumption}[Regularity conditions]\label{assm:regularity}(a) The matrices $n^{-1}X_{\bbS}^\T X_\bbS = \Sigma_n$ for any sample $\mathbb{S}$, and $N^{-1}X_{\bbU}^\T X_{\bbU}=\Sigma_N$ are  positive-definite;
(b) There exists some constant $C$ such that $\|\bx_i\|^2<q C$ for all $i\in\bbU$;
(c)  The finite population is a random sample of a super-population model (\ref{mixed}) satisfying $N^{-1}\sum_{i\in\bbU}y_i^{2+\alpha}<\infty$ for some $\alpha>0$ with $N\rightarrow \infty$.
\end{assumption}

Assumptions \ref{assm:regularity}(a) and (b) are standard regularity conditions related to the auxiliary variables \citep{portnoy1984asymptotic,dai2018broken,chauvet2022asymptotic}. Assumption \ref{assm:regularity}(c) requires the moment conditions to employ the central limit theorem. In contrast to hard calibration, the inexact calibration scheme for
$\bx_{2i}$ involves a correction term on the right-hand side of (\ref{eq:calib-random}), incurring an additional term in $\widehat{\theta}_{w}-\theta_{N}$:
\begin{equation}
    \widehat{\theta}_{w}-\theta_{N}=
N^{-1}\gamma_n(1_{N}^{\T}X_{1,\bbU}D_{12}+1_{N}^{\T}X_{2,\bbU}D_{22})D_{q}^{-1}u + 
N^{-1}\sum_{i\in\bbU}(\delta_i w_i-1)e_i,
\label{eq:extra:term}
\end{equation}
where $\gamma_n$ is considered as a finite-sample tuning parameter for $\gamma$. In $\S$\ref{subsec:Data-adaptive-tuning-parameter}, we propose a data-adaptive approach to select $\gamma_n$ that minimizes the estimated mean squared error of the soft-calibration estimator. 

The following theorem characterizes the asymptotic properties of $\widehat{\theta}_{w}$.
\begin{theorem}\label{thm:theta_linear}
Suppose Assumptions \ref{assum:laten_ignore}-\ref{assm:regularity}, the conditions for $Q(w)$ in Remark \ref{rmk:unique}
hold and $\gamma_n = o(n^{1/2}q^{-1/2})$, the soft-calibration estimator $\widehat{\theta}_{w}$ satisfies
$\widehat{\theta}_{w}-\theta_N=N^{-1}\sum_{i\in\bbU}\psi_{i}(c^{*})-\theta_N+o_{\P}(n^{-1/2})$, where $c^*$ is the solution to $\E\{\partial G(c)/\partial c\mid X_{\bbU},u\}=0$,
\begin{equation}
\psi_{i}(c^{*})=B(c^{*})\bx_{i,\rmsc}+\delta_{i}w(c^{*\T}\bx_{i})\eta_{i}(c^{*}),\quad 
\eta_{i}(c^{*})=y_{i}-B(c^{*})\bx_{i},\label{eq:IF_theta_w}
\end{equation}
$B(c^*) = \left\{
 \sum_{i\in\bbU} \delta_i w'(c^{*\T }x_i)x_iy_i
 \right\}\left\{
 \sum_{i\in \bbU} \delta_i
 w'(c^{*\T }x_i) x_ix_i^{\T}
 \right\}^{-1}$, and $x_{i,\rmsc} = \{x_{1i}^{\T}, x_{1i}^{\T} M_{\bbS} + x_{2i}^{\T}(I_q + R_{\bbS})\}^\T$. As a result, if either the outcome $y_{i}$ follows a linear mixed-effects model or $Q(w)$ entails a correct propensity score model, we have $n^{1/2}(\widehat{\theta}_{w}-\theta_N)\rightarrow N(0,V_{1}+V_{2})$ as $n\rightarrow \infty$,
where
$$
V_{1}=\lim_{n\rightarrow\infty}\frac{n}{N^{2}}\E_{\zeta}\left[{\var}_{\delta}\left\{ \sum_{i\in\bbU}\delta_{i}w(c^{*\T}\bx_{i})\eta_{i}(c^{*})\mid X_{\bbU},\bu,Y_{\bbS}\right\} \mid X_{\bbU}\right],
$$
and 
\[
V_{2}=\lim_{n\rightarrow\infty}\frac{n}{N^{2}}{\var}_{\zeta}\left[\E_{\delta}\left\{ \sum_{i\in\bbU}\psi_{i}(c^{*})\mid X_{\bbU},\bu,Y_{\bbS}\right\} 
\mid X_{\bbU}\right].
\]
\end{theorem}
Theorem \ref{thm:theta_linear} states that $\widehat{\theta}_{w}$ is doubly robust as its consistency requires the outcome following a linear mixed-effects model or the propensity score being correctly specified. We now estimate $V_1$ and $V_2$ by $\widehat{V_1}$ and $\widehat{V_2}$, respectively, in Theorem \ref{thm:var_est}.
\begin{theorem}\label{thm:var_est} Under the assumptions in Theorem \ref{thm:theta_linear},  we have $\widehat{V}_1 = nN^{-2}\sum_{i\in \bbS} w(\widehat{c}^\T\bx_i)^2 {\eta}_i(\widehat{c})^2 \rightarrow V_1$ and $\widehat{V}_2 = 
nN^{-2}\sum_{i\in\bbS} w(\widehat{c}^{\T}x_{i})(y_{i}-x_{1i}^{\T}\widehat{\beta})^2 \rightarrow V_2$ in probability, where $\widehat{\beta}= D_{11}\sum_{i\in\bbS}q_ix_{1i}y_i + 
D_{12}\sum_{i\in\bbS}q_ix_{2i}y_i$
\end{theorem}
Theorem \ref{thm:var_est} estimates $V_{1}$ and $V_2$ by applying the standard variance estimator formula with $c^{*}$ replaced by $\widehat{c}$. As \citet{shao99} show that
the order of $V_{2}/V_{1}$ is $O(n/N)$; thus if the sampling fraction $n/N$ is negligible, we only need to estimate $V_{1}$.

\begin{remark}\label{rmk:bc-estimator}
   In Theorem \ref{thm:theta_linear}, we need $\gamma_n = o(n^{1/2}q^{-1/2})$ to make the bias term (\ref{eq:extra:term}) negligible. If the bias term does not dwindle away, one can use a bias-corrected estimator $\widehat{\theta}_{\dr}$ to correct the remaining bias after soft calibration weighting. Denote $\widehat{\theta}_{\dr}=\widehat{\theta}_w - N^{-1}\sum_{i\in\bbU}\{\delta_iw(\widehat{c}^\T\bx_i)-1\}\widehat{\mu}_i$, which combines soft calibration with the fitted outcomes $\widehat{\mu}_i$ by flexible modeling, similar to \citet{ben2021multilevel} and \citet{Vahe2021High}. 

 As an example, if we combine the soft-calibration estimator with best linear unbiased prediction $\widehat{\mu}_i=x_{1i}^{\T}\widehat{\beta}+x_{2i}^{\T}\widehat{u}$, $\gamma_n$ is allowed to grow faster with $n$ than requested in Theorem \ref{thm:theta_linear} under the linear mixed-effects model, implying that $\widehat{\theta}_{\dr}$ is more robust than $\widehat{\theta}_w$ against the rate requirement for $\gamma_n$. Other choices for outcome models can also effectively reduce the left-over bias as long as they can approximate the true outcome $\E_{\zeta}(y_i\mid x_i)$ well enough. A detailed discussion of its asymptotic properties is deferred to $\S$\ref{sec:proof_rmk_dr} of the Supplementary Material. 
\end{remark}

\subsection{Data-adaptive tuning parameter selection\label{subsec:Data-adaptive-tuning-parameter}}

To properly choose the tuning parameter $\gamma_n$, we propose a data-adaptive cross-fitting strategy that targets minimizing the mean squared error of the soft-calibration estimator $\widehat{\theta}_{w}$. Specifically, we divide the data into $\mathcal{B}$ disjoint groups $\mathcal{I}_b, b=1,\cdots, \mathcal{B}$. Let $\widehat{c}_{-k}$ and $\widehat{\beta}_{-k}$ denote the estimator of $c^*$ and $\beta$ computed using the observations from all the folds except the $k$-th fold based on the soft conditions with the tuning parameter $\gamma_n$. The estimated mean squared error will be
\begin{align*}
    &\text{MSE}(\widehat{\theta}_w;\gamma_n)=
    \frac{1}{\mathcal{B}}
    \sum_{k=1}^{\mathcal{B}}
    \left[
    \left\{
    \frac{\mathcal{B}}{N}
    \sum_{i\in \mathcal{I}_k}\delta_i w(\widehat{c}_{-k}^\T x_i)y_i\right\}-\theta_N
    \right]^2\\
    &+\frac{1}{\mathcal{B}}\sum_{k=1}^{\mathcal{B}}
    \frac{\mathcal{B}^2}{N^2}
    \left[
    \sum_{i\in\mathcal{I}_k}
    \delta_i w(\widehat{c}_{-k}^\T x_i)^2\{y_i-B(\widehat{c}_{-k})x_i\}^2 +
    \sum_{i\in\mathcal{I}_k}
    \delta_i w(\widehat{c}_{-k}^\T x_i)(y_i-x_{1i}^\T \widehat{\beta}_{-k})^2
    \right],
\end{align*}
where the unknown parameter $\theta_N$ is approximated by the hard-calibration estimator $\widehat{\theta}_{\rm hc}$ as a proxy. Given this cross-fitting scheme, $\text{MSE}(\widehat{\theta}_w;\gamma_n)$ is able to approximate the true mean squared error with negligible bias. A similar strategy
has been used by \citet{xiao2013comparison} for tuning parameter selection in other
contexts. We select $\gamma_n$ by minimizing the estimated mean squared error of $\widehat{\theta}_{w}$
over a discrete grid $\{\gamma_n^{*}\times10^{j}:j=-5,\ldots,5\}$,
where $\gamma^{*}_n$ is a user-provided value. Our tuning strategy involves specifying $\gamma_n^*$ and one candidate can be $\widehat{\sigma}_e^2/\widehat{\sigma}_u^2$, where $\widehat{\sigma}_e^2$ and $\widehat{\sigma}_u^2$ 
are the restricted maximum likelihood estimators of ${\sigma}_{e}^{2}$
and ${\sigma}_{u}^{2}$, respectively \citep{golub1979generalized}.

\subsection{Cluster-specific Nonignorable Missingness\label{subsec:Cluster-specific-nonignorable-mi}}

We now consider one important extension of latent ignorability to cluster-specific nonignorable missingness, and another extension to causal inference in the presence of unmeasured cluster-level confounders is presented in $\S$\ref{sec:sim_res_causal}. Following the conventional notations for clustered data, consider
the finite population $\mathcal{F}_{N}=\{(\bx_{ij},y_{ij},\delta_{ij}):i=1,\ldots,K,j=1,\ldots,N_{i}\}$,
where $i$ indexes the cluster and $j$ indexes the unit within each cluster,
$y_{ij}$ is the outcome of interest for the $j$-th unit in cluster
$i$, which is subject to missingness, $\bx_{ij}\in\mathbb{R}^{p}$
is the vector of observed covariates, $\delta_{ij}$ is the response
indicator with value one if $y_{ij}$ is observed and zero otherwise,
and $N=\sum_{i=1}^{K}N_{i}$ is the population size. The parameter
of interest is $\theta_N=N^{-1}\sum_{i=1}^{K}\sum_{j=1}^{N_{i}}y_{ij}$.
We consider the two-stage cluster sampling: in the first stage, $k$ clusters
are selected from $K$ clusters with cluster sampling weights $d_{i}$,
and in the second stage, a random sample of $n_{i}$ units is selected
from each sampled cluster $i$ with unit sampling weights $N_{i}/n_{i}$.
The sample size is $n=\sum_{i=1}^{k}n_{i}$. Assume the outcome follows
the linear mixed-effects model 
\[
y_{ij}=\bx_{ij}^{\T}\beta+a_{i}+e_{ij}=\bx_{ij}^{\T}\beta+z_{ij}^{\T}a+e_{ij},\quad i=1,\ldots,k,\quad j=1,\ldots,n_{i},
\]
where $a=(a_{1},\ldots,a_{k})^{\T}$ are the latent
cluster-specific random effects, and $z_{ij}=s_{i}$
with $s_{i}$ being the canonical coordinate basis for
$\mathbb{R}^{k}$ as the cluster indicator. Here, $\bx_{ij}$, $z_{ij}$
and $a$ are the counterparts of $\bx_{1i},\bx_{2i}$
and $\bu$ in $\S$\ref{sec:Basic-setup}.

In the presence of missing data, the sample average of the observed
$y_{ij}$ even adjusted for sampling design weights may be biased
for $\theta_N$ due to the selection bias associated with the respondents. To correct such selection bias, the calibrated propensity score method proposed by \citet{kim2016b} imposes the following hard calibration constraints for both fixed effects and cluster effects:
\begin{equation}
\sum_{i=1}^{k}\sum_{j=1}^{n_{i}}d_{ij}\delta_{ij}w_{ij}\bx_{ij}=\sum_{i=1}^{k}\sum_{j=1}^{n_{i}}d_{ij}\bx_{ij},\label{7}
\end{equation}
and $\sum_{j=1}^{n_{i}}d_{ij}\delta_{ij}w_{ij}=\sum_{j=1}^{n_{i}}d_{ij}$ for $i=1,\ldots,k$ with $d_{ij}=d_iN_in_i^{-1}$. The calibration constraints for the cluster effects may be stringent when the clusters weakly affect the outcome and may be relaxed to the following under soft calibration
\begin{align}
 & \sum_{i=1}^{k}\sum_{j=1}^{n_{i}}d_{ij}\delta_{ij}w_{ij}=\sum_{i=1}^{k}\sum_{j=1}^{n_{i}}d_{ij},\label{eq:hard_cal_survey_int}\\
 & \sum_{i=1}^{k}\sum_{j=1}^{n_{i}}d_{ij}\delta_{ij}w_{ij}z_{ij}=\sum_{i=1}^{k}\sum_{j=1}^{n_{i}}d_{ij}z_{ij}+\sum_{i=1}^{k}\sum_{j=1}^{n_{i}}d_{ij}M_{\bbS}^{\T}\bx_{ij}+\sum_{j=1}^{n_{i}}d_{ij}R_{\bbS}^{\T}z_{ij},\label{eq:soft_cal_survey}
\end{align}
where (\ref{eq:hard_cal_survey_int}) is still an exact constraint
forcing the weighted estimator of the population size to be the same
as the design-weighted estimator, and (\ref{eq:soft_cal_survey})
is an approximate calibration for cluster effects. The adjustment in (\ref{eq:soft_cal_survey}) relaxes the requirement of an exact calibration of cluster effects, which can be beneficial when the outcome has relatively homogeneous cluster-specific effects, that is, the ratio $\sigma_e^2/\sigma_u^2$ is large. Thus, our soft-calibration estimator of $\theta_N$ is
$\widehat{\theta}_{w}=N^{-1}\sum_{i=1}^{k}\sum_{j=1}^{n_{i}}d_{ij}\delta_{ij}w(\widehat{c}^\T \bx_{ij})y_{ij}$, where $w(\widehat{c}^\T \bx_{ij})$ is obtained
by minimizing a given loss function subject to the soft calibration
constraints (\ref{7}), (\ref{eq:hard_cal_survey_int}) and (\ref{eq:soft_cal_survey}). 

\begin{corollary}\label{coro_theta_cluster}
Under Assumptions \ref{assum:laten_ignore}(a), \ref{assm:regularity}, other regularity conditions in Assumption \ref{assum:regular-sample} of the Supplementary Material, and $\gamma_n = o(n^{1/2}q^{-1/2})$, if either the outcome $y_{ij}$ follows a linear mixed-effects model or $Q(w)$ entails a correct propensity score model, we have $n^{1/2}(\widehat{\theta}_{w}-\theta_N)\rightarrow  N(0,V_{1})$
as $n\rightarrow\infty$ and $n/N\rightarrow f\in[0,1)$, where $V_{1}=\lim_{n\rightarrow\infty}nN^{-2}\var_{p}\left\{ \sum_{i=1}^{k}d_{i}\psi_{i}(c^{*})\mid \mathcal{F}_N\right\}$,
\[
\psi_{i}(c^{*})=\frac{N_{i}}{n_{i}}\sum_{j=1}^{n_{i}}\left\{ B(c^{*})\bx_{ij,\rmsc}+\delta_{ij}w(c_0^{*\T}x_{ij}+c_1^{*\T}z_{ij})\eta_{ij}(c^{*})\right\},\quad c^* = (c_0^{*\T},c_1^{*\T})^{\T},
\]
and $\eta_{ij}(c^{*})=y_{ij}-B(c^{*})
(\bx_{ij}^{\T}, z_{ij}^{\T})^{\T}$ with 
$\var_p(\cdot)$ being the variance under the
clustered sampling design and $\{B(c^{*}),\bx_{ij,\rmsc}\}$ defined in $\S$\ref{sec:proof of corollary_1} of the Supplementary Material.
\end{corollary}

The results in Corollary \ref{coro_theta_cluster} are similar to
that of Theorem \ref{thm:theta_linear} except that $V_{2}$ under
two-stage cluster sampling is negligible compared to $V_{1}$ even
though $n/N$ or some cluster sampling fractions $n_{i}/N_{i}$ are
not negligible \citep{shao99} and thus is omitted. For variance estimation, the variance of
$\widehat{\theta}_{w}$ can be consistently estimated as $\widehat{V}_{1}=nN^{-2}\sum_{i=1}^{k}\sum_{j=1}^{k}\Omega_{i,j}\psi_{i}(\widehat{c})\psi_{j}(\widehat{c})$, where $\Omega_{i,j}$ depends on the cluster sampling scheme at the
first stage, $\psi_{i}(\widehat{c})$ is referred as
the pseudo-values with $c^{*}$ replaced by $\widehat{c}$,
and the consistency of $\widehat{V}_{1}$ can be verified by standard
arguments in \citet{kim2009}.

\section{Simulation study\label{sec:Simulation}}

In this section, we conduct a simulation study to evaluate the finite-sample performance of our proposed soft-calibration estimator and assess the robustness of its bias-corrected version in the case of cluster-specific nonignorable missingness. First, we generate samples from finite populations using the two-stage cluster sampling mechanism, in which $k=30$ clusters with cluster sizes $n_{i}=200$ are selected from $K=2000$ clusters. 

We consider two generating models for $y_{ij}$. One is the linear mixed-effects model: $y_{ij}=x_{ij}^{\T}{\beta}+\lambda_{1}a_{i}+e_{ij}$ with $x_{ij}=(1,x_{1ij},x_{2ij})^{\T}$ where ${\beta}=(0,1,1)^{\T}$, $x_{1ij}\sim U[-0.75,0.75]$, $x_{2ij}\sim N(0,1)$, $a_{i}\sim N(0,1)$ and $e_{ij}\sim N(0,1)$. The other one is a non-linear mixed-effects model $y_{ij}=x_{ij}^{\T}\beta + x_{1ij}^2 + x_{2ij}^2 + 0.1x_{3ij}^{\dagger}+ 0.1x_{4ij}^{\dagger}+\lambda_{1}a_{i}+e_{ij}$, where $x_{3ij}^{\dagger}$ and $x_{4ij}^{\dagger}$ are the standardized versions of $x_{3ij}=\exp(x_{1ij})$ and $x_{4ij}=\exp(x_{2ij})$. We consider a logistic propensity score to generate $\delta_{ij}$:
$\delta_{ij}\sim\text{Bernoulli }(p_{ij})$, where $\text{logit}(p_{ij})=\bx_{ij}^{\T}\alpha+\lambda_{2}z_{i}$
and $\alpha=(-0.25,1,1)^{\T}$ with $\text{logit}(\cdot)$ being the logit-link. For illustration, we present a set of $(\lambda_1,\lambda_2)$ in Table \ref{tab:sim-survey:all} gauging the between-cluster variation of $y_{ij}$ and $\delta_{ij}$, and additional simulation studies are deferred to $\S$\ref{sec:additional_sims} of the Supplementary Material.

From $\S$\ref{subsec:soft_cal_loss}, the loss function dictates the propensity score model. For assessing the double robustness of the soft-calibration estimator, we consider two loss functions: the maximum entropy balancing function, i.e., a logistic mixed-effects model for the propensity score, and the square loss function, i.e., a linear mixed-effects model for the inverse of the propensity score. Next, we compute nine estimators for $\theta_N$: (i)  $\widehat{\theta}_{\rm sim}$ the simple average of the observed $y_{ij}$; (ii, iii) $\widehat{\theta}_{\rm fix}$ and $\widehat{\theta}_{\rm rand}$, where $p_{ij}$ is estimated with fixed or random effects for clusters; (iv-vi) $\widehat{\theta}_{\rm hc}$, $\widehat{\theta}_w^{\sq}$ and $\widehat{\theta}_{w}^{\me}$, where $w_{ij}$ achieves the hard calibration conditions under the maximum entropy loss function, the soft calibration conditions under the square loss function or under the maximum entropy loss function; (vii) $\widehat{\theta}_{\dr}$, bias-corrected $\widehat{\theta}_{w}^{\me}$ with $\widehat{\mu}_{ij}=x_{1ij}^{\T}\widehat{\beta}+x_{2ij}^{\T}\widehat{u}$; (viii)
$\widehat{\theta}_{\rm cbps}$, the high-dimensional covariate propensity score balancing method of \cite{ning2020robust}; and (ix) $\widehat{\theta}_{\rm rcal}$, the high-dimensional regularized calibration method of \cite{tan2020model}.

\begin{longtable}{lllllllllllll}
\caption{Bias $(\times 10^{-2})$, variance $(\times10^{-3})$, mean squared error
$(\times10^{-3})$ and coverage probability (\%) of the estimators under cluster-specific nonignorable missingness based on $500$ simulated datasets
\label{tab:sim-survey:all}}\\
\hline
&& 
$\widehat{\theta}_{\rm sim}$ &
$\widehat{\theta}_{\rm fix}$ & 
$\widehat{\theta}_{\rm rand}$ &
$\widehat{\theta}_{\rm hc}$ &
$\widehat{\theta}_{w}^{\sq}$ &
$\widehat{\theta}_{w}^{\me}$ &
$\widehat{\theta}_{\dr}$&
$\widehat{\theta}_{\rm cbps}$ & $\widehat{\theta}_{\rm rcal}$ \\ 
\hline
\hline
\multicolumn{10}{l}{\textit{Linear mixed-effects model with $(\lambda_1,\lambda_2)=(0.01,1)$}}\\
\hline
Bias&&21.2&0.02&0.29&0.10&0.16&0.13&0.09&0.17&0.35\\
VAR&&0.23&1.53&1.40&0.78&0.61&0.73&0.74&0.78&0.75\\
MSE&&45.1&1.53&1.41&0.78&0.61&0.73&0.74&0.78&0.76\\
CP&&0.0&94.6&94.2&92.6&93.8&93.0&93.2&--&--\\
\\
\multicolumn{10}{l}{\textit{Linear mixed-effects model with $(\lambda_1,\lambda_2)=(0.01,10)$}}\\
\hline
Bias&&5.02&0.28&0.01&0.73&0.29&0.27&0.18&0.43&7.44\\
VAR&&0.35&26.4&22.3&4.57&1.49&1.69&2.16&5.88&0.69\\
MSE&&2.88&26.4&22.3&4.62&1.49&1.70&2.16&5.89&6.23\\
CP&&23.8&88.6&87.8&94.2&94.4&92.4&92.2&--&--\\
\\
\multicolumn{10}{l}{\textit{Linear mixed-effects model with $(\lambda_1,\lambda_2)=(0.5,1)$}}\\
\hline
Bias&&30.3&0.49&1.61&0.64&1.26&1.28&0.63&0.82&2.03\\
VAR&&2.74&10.7&10.2&9.23&9.64&9.84&9.21&10.2&9.79\\
MSE&&94.4&10.7&10.4&9.27&9.80&10.0&9.25&10.3&10.2\\
CP&&0.0&95.0&93.4&94.2&94.0&93.6&94.0&--&--\\
\\
\multicolumn{10}{l}{\textit{Non-linear mixed-effects model with $(\lambda_1,\lambda_2)=(0.01,1)$}}\\
\hline
Bias&&31.6&0.10&0.38&0.92&8.75&0.03&0.11&0.09&0.59\\
VAR&&1.50&2.42&2.24&1.96&1.72&1.69&1.71&1.71&1.86\\
MSE&&102&2.42&2.25&2.05&9.37&1.69&1.71&1.71&1.89\\
CP&&0.0&94.0&94.0&92.6&0.0&94.4&96.6&--&--\\
\hline
\end{longtable}

Table \ref{tab:sim-survey:all} reports the simulation results based on $500$ Monte Carlo samples. The performance of estimators is evaluated on the basis of biases, variances, mean squared errors, and coverage probabilities. Among all estimators, the simple average estimator $\widehat{\theta}_{\rm sim}$ shows large biases across all different scenarios. When the cluster factor is included as fixed or random effects, the biases of $\widehat{\theta}_{\rm fix}$ and $\widehat{\theta}_{\rm rand}$ are substantially
reduced, while their variances remain large. The large variances could be attributed to their overly abundant parameters associated with the cluster indicators. When the random effects weakly affect outcomes (i.e., $\lambda_1=0.01$), all soft-calibration estimators outperform $\widehat{\theta}_{\rm hc}$, indicating their ability to address the issue of over-calibration. In particular, $\widehat{\theta}_{w}^{\sq}$ performs better than $\widehat{\theta}_{w}^{\me}$ under the linear mixed-effects model, which agrees with the connection between $\widehat{\theta}_{w}^{\sq}$ and $\widehat{\theta}_{\rm blup}$ established in Proposition \ref{prop:BLUP-penalized}. However, $\widehat{\theta}_{w}^{\sq}$ is subject to significant bias when the outcome model is misspecified, leading to an unsatisfactory coverage probability, while $\widehat{\theta}_{w}^{\me}$ still exhibits a desirable finite-sample coverage probability, which aligns with our claim of double robustness in Theorem \ref{thm:theta_linear} when the propensity score is correctly specified. Although the bias-corrected estimator $\widehat{\theta}_{\dr}$ has a slightly larger mean squared error than $\widehat{\theta}_w^{\me}$ when $\lambda_1=0.01$, it performs better when the data present a larger between-cluster variation of $y_{ij}$ (i.e., $\lambda_1=0.5$), which provides empirical support for the robustness of $\widehat{\theta}_{\dr}$ with respect to the rate requirement for $\gamma_n$. As expected, both regularized calibration estimators $\widehat{\theta}_{\rm cbps}$ and $\widehat{\theta}_{\rm rcal}$ have larger mean squared errors under the linear mixed-effects model since our soft calibration conditions are motivated by linear mixed effects.

Overall, our proposed estimators tend to produce smaller mean squared errors while dealing with cluster-specific missingness, irrespective of possible model misspecification of either outcome or propensity score.

\section{An application: effect of school-based BMI screening on childhood obesity \label{sec:Application}}

The epidemic of childhood obesity has been widely publicized \citep{peyer2015factors}. Many school districts have implemented coordinated school-based body mass index screening programs to help increase parental awareness of children's body status and promote preventive strategies to reduce the risk of obesity. We use a data set collected by the Pennsylvania Department of Health to evaluate the effect of the program on the annual prevalence of overweight and obesity in elementary schools across Pennsylvania in $2007$. The primary goal is to investigate
the causal effect of implementing the program on reducing childhood obesity and overweight. Because the implementation of the policy was
not randomized, it is essential to control pre-treatment covariates
for causal analysis of the effect of the policy. Furthermore, school districts are clustered by geographic and demographic factors. Thus, soft calibration can be used to estimate the causal effect by correcting for cluster-specific confounding
bias. 

The dataset contains information on $493$ elementary schools,
which are clustered according to the type of community (rural, suburban,
and urban) and the population density (low, moderate, and high). There
are six clusters with sample size $n_{1}=65,n_{2}=96,n_{3}=89,n_{4}=29$,
$n_{5}=104$, and $n_{6}=4$. For each school, the data consist of
the treatment status $A_{ij}$ where $A_{ij}=1$ if the school has implemented
the policy and $0$ otherwise, the outcome variable $y_{ij}$, indicating the annual prevalence of overweight and obesity in each school, and two covariates $x_{1ij}$ and $x_{2ij}$, the baseline prevalence of overweight children and the
percentage of reduced and free lunch, respectively. For estimation,
we consider the linear mixed-effects model and the maximum entropy loss function, including covariates $x_{1ij}$, $x_{2ij}$ and the cluster intercept to model the outcome and weights for $A_{ij}=0$ and $A_{ij}=1$, respectively.

Table \ref{tab:real} reports the estimated average treatment effects on the annual prevalence of overweight and obesity along with the estimated variances and $95\%$ confidence intervals. Without any
adjustment, $\widehat{\theta}_{\rm sim}$ shows that the policy has a significant effect in reducing the prevalence of overweight and
obesity in schools, which may be subject to confounder bias. After adjusting for
confounders through propensity weighting or calibration, all other estimators show that the policy may mildly reduce the prevalence of overweight. Also, $\widehat{\theta}_{\rm hc}$, $\widehat{\theta}_{w}^{\me}$ and $\widehat{\theta}_{\dr}$ provide similar estimates, but the soft-calibration estimators yield a slightly smaller variance, which can be attributed to the approximate calibration condition on the cluster indicator. As discussed in $\S$\ref{sec:additional_app} of the Supplementary Material, the cross-fitting strategy selects two small tuning parameters as $\gamma_{n, A=0}=0.052$ and $\gamma_{n, A=1}=0.068$. It implies that the correction term on the right-hand side of (\ref{eq:calib-random}) is fairly small and a nearly exact calibration should be adopted, as demonstrated by the similarities in the calibration weights in Figure \ref{fig:weights}. Estimators $\widehat{\theta}_{\rm cbps}$ and $\widehat{\theta}_{\rm rcal}$ might not be credible when the sparsity condition is not met, as we have shown in the simulation studies. Based on our analysis, the policy can reduce the average prevalence of overweight and obesity in elementary schools in Pennsylvania, although the statistical evidence is not significant. 

\begin{table}[htbp]
\caption{The estimated average treatment effects of SBMIS on the annual overweight and obesity prevalence in elementary schools across Pennsylvania\label{tab:real}}
\vspace{0.15cm}
\resizebox{\textwidth}{!}{%
\begin{tabular}{lccccccccccc}
\toprule
 & 
$\widehat{\theta}_{\rm sim}$ &
$\widehat{\theta}_{\rm fix}$ & 
$\widehat{\theta}_{\rm rand}$ &
$\widehat{\theta}_{\rm hc}$ & 
$\widehat{\theta}_{w}^{\me}$ &
$\widehat{\theta}_{\dr}$&
$\widehat{\theta}_{\rm cbps}$ & $\widehat{\theta}_{\rm rcal}$ & \\
\midrule
ATE &8.71& 0.41 & 0.43 & 0.55 & 0.53 &0.54& 0.28 & 0.51 & \\
VE & 2258.8 & 467.8 & 474.5 &448.5 & 445.7 & 446.0 &  & \\
CIs & (5.77,11.66) & (-0.93,1.75) & (-0.92,1.78) & (-0.77,1.86) & (-0.78,1.84) & (-0.77,1.85) & -- & --&\\
\bottomrule
\end{tabular}}

\begin{tablenotes}
\item SBMIS, school-based body mass index screening; ATE, average treatment effects; VE, variance estimation $(\times 10^{-3})$; CIs, confidence intervals
\end{tablenotes}
\end{table}

\section*{Acknowledgement}
This research is supported by the U.S. National Science Foundation and the U.S. National Institutes of Health.

\bibliographystyle{dcu}
\bibliography{ref}
\setcounter{equation}{0}
\setcounter{table}{0}
\setcounter{theorem}{0}
\setcounter{lemma}{0}
\setcounter{figure}{0}
\renewcommand{\theequation}{S\arabic{equation}}
\renewcommand{\thetable}{S\arabic{table}}  
\renewcommand{\thefigure}{S\arabic{figure}}
\renewcommand{\thesection}{S\arabic{figure}}

\setcounter{algocf}{0} 
\makeatletter
\renewcommand{\thealgocf}{S\@arabic\c@algocf}
\makeatother
\renewcommand\thetheorem{{S}\arabic{theorem}}
\renewcommand\thelemma{{S}\arabic{lemma}}
\renewcommand\thecorollary{{S}\arabic{corollary}}
\renewcommand\theproposition{{\rm S}\arabic{proposition}}
\renewcommand\thedefinition{{\rm S}\arabic{definition}}
\renewcommand\theassumption{{ S}\arabic{assumption}}
\renewcommand\theremark{{S}\arabic{remark}}
\renewcommand\thestep{{\rm S}\arabic{step}}
\renewcommand\thecondition{{\rm S}\arabic{condition}}
\renewcommand\theexample{{\rm S}\arabic{example}}

\newpage
\appendix
\section*{Supplementary material}
\section{Technical Proofs and Details}
\subsection{Summary}\label{sec:proof_summary}
We include all the technical details in this section. In specific, $\S$\ref{subsec:proof_prop1} provides the proof for Proposition \ref{prop:BLUP-penalized}. $\S$\ref{subsec:proof_unique} provides the proof for Remark \ref{rmk:unique}. $\S$\ref{subsec:proof_thm1} provides the proof for Theorem \ref{thm:theta_linear}. $\S$\ref{sec:proof of corollary_1} provides the proof for Corollary \ref{coro_theta_cluster}. $\S$\ref{sec:nonignore_trt} provides technical details and proofs for soft calibration under cluster-specific nonignorable treatment mechanism. $\S$\ref{sec:implement_details} provides the implementation details of the soft-calibration estimator. $\S$\ref{sec:proof_rmk_dr} provides technical details and proofs for the bias-corrected estimator. Table \ref{tab:notations_additional} introduces additional notations which will be used throughout this Supplementary Material. 
\begin{table}[!ht]
\caption{\label{tab:notations_additional}Summary of the notations}
\vspace{0.15cm}
{
      \begin{tabular}{ll}
      \toprule
          Notation & Definition \\
        \midrule
          $\E_{e}(\cdot), \E_{u}(\cdot)$&
          Expectations with respect to the random error $e$ and the random effect $u$\\
          $\var_{e}(\cdot), \var_{u}(\cdot)$&
          Variances with respect to the random error $e$ and the random effect $u$\\
          $\E_{p}(\cdot), \var_p(\cdot)$&
          Expectation and variance with respect to the cluster-level sampling design\\
          $\mathbb{W}$, $I(\mathbb{W})$&
          The support and image sets of the objective function $Q(w)$ with respect to $w$\\
          $\mathbb{C}$, $I(\mathbb{C})$&
          The support and image sets of the weight function $w(c^{\T}x)$ with respect to $c$ \\
          $\Sigma_q$& $\Sigma_q=\text{diag}(q_1,\cdots,q_n)$ with known factor $\{q_i:i\in\bbS\}$ for unequal variance\\
          $g(z)$& $g(z)=z\cdot (Q')^{-1}(z)-Q\{(Q')^{-1}(z)\}$ is the convex conjugate function of $Q(w)$\\
          $\|\cdot\|$&  $\|\cdot\|$ is the $L_2$ norm of a vector and the spectral norm of a matrix\\
          $T_r$, $T_X$&
          $T_r=N^{-1}\sum_{i\in\bbU}(\bx_{1i}^{\T}M_{\bbS} + \bx_{2i}^{\T}R_{\bbS})$, $ T_X= N^{-1}(\sum_{i\in\bbU}x_{1i}^{\T},\sum_{i\in\bbU}x_{2i}^{\T}+NT_r)^{\T}$\\
          \bottomrule
      \end{tabular}}
  \end{table}

\subsection{Proof of Proposition \ref{prop:BLUP-penalized}}\label{subsec:proof_prop1}
\begin{proof}
The solution to (\ref{eq:soft_sq}) can be formally derived by the means of Lagrange multiplier:
$$\sigma_{e}^{2}(w-1_{n})^{\T}\Sigma_q^{-1}(w-1_{n})+\sigma_{u}^{2}(w^{\T}X_{2,\bbS}-1_{N}^{\T}X_{2,\bbU})D_{q}(w^{\T}X_{2,\bbS}-1_{N}^{\T}X_{2,\bbU})^{\T}-2(w^{\T}X_{1,\bbS}-1_{N}^{\T}X_{1,\bbU})\lambda.$$
By taking the first derivative with respect to $w$ and setting it to zero, we have
\begin{align*}
&\widehat{w}^{\sq} =V_{\bbS}^{-1}(X_{1,\bbS}\widehat{\lambda}+\sigma_{e}^{2}\Sigma_q^{-1}+\sigma_{u}^{2}X_{2,\bbS}D_{q}X_{2,\bbU}^{\T}1_{N}),\\
&\widehat{\lambda}=(X_{1,\bbS}^{\T}V_{\bbS}^{-1}X_{1,\bbS})^{-1}\{X_{1,\bbU}^{\T}1_{N}-X_{1,\bbS}^{\T}V_{\bbS}^{-1}(\sigma_{e}^{2}\Sigma_q^{-1}+\sigma_{u}^{2}X_{2,\bbS}D_{q}X_{2,\bbU}^{\T}1_{N})\}.
\end{align*}
It can be shown that 
$$
\begin{aligned}
\widehat{w}^{\sq}&=1_{n}+\Sigma_qX_{\bbS}\{X_{\bbS}^{\T}\Sigma_qX_{\bbS}+\gamma\text{diag}(0_{p},D_{q}^{-1})\}^{-1}(X_{\bbU}^{\T}1_{N}-X_{\bbS}^{\T}1_{n}).
\end{aligned}$$
Furthermore, the soft-calibration estimator can be expressed by
\begin{align*}\widehat{w}^{\sq\T}Y_{\bbS}&=
1_{n}^{\T}Y_{\bbS}
+
(1_{N}^{\T}X_{\bbU}-1_{n}^{\T}X_{\bbS})\{X_{\bbS}^{\T}\Sigma_qX_{\bbS}+\gamma\text{diag}(0_{p},D_{q}^{-1})\}^{-1}X_{\bbS}^{\T}\Sigma_q Y_{\bbS}\\
&=1_{n}^{\T}Y_{\bbS}+(1_{N}^{\T}X_{1,\bbU}-1_{n}^{\T}X_{1,\bbS})\widehat{\beta}+(1_{N}^{\T}X_{2,\bbU}-1_{n}^{\T}X_{2,\bbS})\widehat{\bu}\\
&=1_{N}^{\T}X_{1,\bbU}\widehat{\beta}+1_{N}^{\T}X_{2,\bbU}\widehat{\bu},
\end{align*}
where $(\widehat{\beta}, \widehat{\bu})$ is the solution to the mixed-model equations, and $1_n^{\T}Y_{\bbS}=1_n^{\T}X_{1,\bbS}\widehat{\beta}+1_n^{\T}X_{2,\bbS}\widehat{\bu}$ as long as $X_{1,\bbS}$ contains an intercept. Therefore, the desired result follows, that is, $\widehat{\theta}_{{\rm blup}} =\widehat{\theta}_{w}^{\sq}$. 
\end{proof}

\subsection{Proof of Remark \ref{rmk:unique}}\label{subsec:proof_unique}
Variants of Remark \ref{rmk:unique} have been proved before, the proof provided here under our setting is mostly adapted from \cite{deville1992calibration}, Result 1.
\begin{proof}
Suppose $q(w)=Q'(w)$ is normalized such that $q(1)=0$ and $q'(1)=1$,
$q(w_i) - c^\T \bx_i = 0$, $w_i = q^{-1}(c^\T \bx_i) = w(c^\T \bx_i)$, where $w(\cdot)$ is the reciprocal mapping of $q(\cdot)$ defined as $q:\mathbb{W}\rightarrow I(\mathbb{W})$ and $w :I(\mathbb{W})\rightarrow \mathbb{W}$, where $I(\mathbb{W})$ is the image of the function $Q(\cdot)$. Therefore, $q(1)=0$ implies $w(0)=1$, $q'(1)=1$ implies $w'(0) = 1/q'(1)=1$. Since $Q(w_i)$ is strictly convex, implying $q', w'>0$, $q,w$ are both strictly increasing functions. Define $\Upsilon: \mathbb{C}\rightarrow I(\mathbb{C})$ as $\Upsilon_n(c) = N^{-1}
\sum_{i\in\bbU}\{\delta_i w(c^\T \bx_i)-1\}\bx_i$, where $\mathbb{C} = \cap_{i\in\bbU} \{c: c^\T x_i\in I(\mathbb{W})\}$. Let $\mathbb{C}^*$ be an open convex set of $\mathbb{C}$, $\Upsilon_N(c) = \mathbb{E}\{\Upsilon_n(c)\mid X,\bu\}$ is well-defined for any $c\in\mathbb{C}^*$, and $\Upsilon_n$ converges to $\Upsilon_N$ under Assumption \ref{assum:laten_ignore}. 

Note the property that $\Upsilon'_N(c) = N^{-1}\sum_{i\in \bbU} \pi_i w'(c^\T \bx_i)\bx_i\bx_i^\T$ is positive definite by Assumption \ref{assm:regularity}(a) for all $\bu$. $\Upsilon_N(c)$ is injective function and maps $\mathbb{C}^*$ onto an open set $I(\mathbb{C}^*)$. By Assumption \ref{assum:positivity} and $w$ is strictly increasing function, there exist $c^*\in\mathbb{C}$ s.t. $Nn^{-1}\underline{d}<w(c^{*\T} \bx_i)<Nn^{-1}\overline{d}$ and $E\{\Upsilon_N(c^*)\mid X\}=0$. Therefore, the image set of $\Upsilon_N(\cdot)$ contains a closed sphere $I_r(\mathbb{C})\subset I(\mathbb{C})$ with radius $r$ in the neighborhood of $0_{p+q}$. 
Let $\mathbb{C}^*_r$ be the compact set $\Upsilon^{-1}_N(I_r(\mathbb{C}))$, and therefore $\Upsilon_n(\cdot)$ has a inverse function  on $I_r(\mathbb{C})$ with probability 1 and $N^{-1}\sum_{i\in\bbU}(R_{\bbS}^{\T}\bx_{2i} + M_{\bbS}^{\T}\bx_{1i}) = O_\P(\gamma_n  n^{-1}{q}^{1/2})$, which belongs to $I_r(\mathbb{C})$ with probability 1 when the radius $r$ is large enough.  As the (\ref{eq:est_alpha_beta}) can be written as $\Upsilon_n(c) = \{0_{p\times 1}, N^{-1}\sum_{i\in\bbU}(\bx_{2i}^{\T}R_{\bbS} + \bx_{1i}^{\T}M_{\bbS})\}^{\T}$, the equation has a unique solution with probability 1 followed by the injective nature of $\Upsilon_n(c)$.
\end{proof}

\subsection{Proof of Theorem \ref{thm:theta_linear}}\label{subsec:proof_thm1}
\begin{proof}
Let $U(c)=\partial G(c)/\partial c$. By standard Taylor's Theorem, 
\begin{align*}
    \widehat{\theta}_{w}(\widehat{c})&=\widehat{\theta}_{w}(c^{*})+
    \left\{\partial\widehat{\theta}(\tilde{c})/\partial c^{\T}\right\}
    (\widehat{c}-c^{*})\\
    &=\widehat{\theta}_{w}(c^{*})-\left\{\partial\widehat{\theta}(\tilde{c})/\partial c^{\T}\right\} \left\{
    \frac{\partial^2 G(c^*)}{\partial c\partial c^\T}\right\}^{-1} \frac{\partial G(c^*)}{\partial c^*}\\
    &=\widehat{\theta}_{w}(c^{*})-N^{-1}\left\{ \sum_{i\in\bbU}\delta_{i}w'(\tilde{c}^{\T}\bx_{i})\bx_{i}y_{i}\right\} \left\{ \sum_{i\in\bbU}\delta_{i}w'(\overline{c}^{\T}\bx_{i})\bx_{i}\bx_{i}^{\T}\right\} ^{-1} U(c^*),
\end{align*}
where $U(c)=\partial G(c)/\partial c$, $\tilde{c}$ and $\overline{c}$ lie on the line joint $\widehat{c}$ and $c^{*}$. 
\begin{lemma}\label{lm:c_hat}
$\|\widehat{c}-c^*\|=O_\P(n^{-1/2})$.
\end{lemma}
By Lemma \ref{lm:c_hat} and $w'(\cdot)$ is continuous function, we have 
\begin{align*}
    \widehat{\theta}_w(\widehat{c})&= \widehat{\theta}_{w}(c^{*})-N^{-1}\left\{ \sum_{i\in\bbU}\delta_{i}w'(c^{*\T}\bx_{i})\bx_{i}y_{i}\right\} \left\{ \sum_{i\in\bbU}\delta_{i}w'(c^{*\T}\bx_{i})\bx_{i}\bx_{i}^{\T}\right\} ^{-1} U(c^{*})+o_{\P}(n^{-1/2})\\
    & =\widehat{\theta}_{w}(c^{*})+N^{-1}\Gamma(c^{*})i(c^{*}) U(c^{*})+o_{\P}(n^{-1/2}),
\end{align*}
with $\Gamma(c) = \sum_{i\in\bbU}\delta_iw'(c^{\T}\bx_i)x_iy_i$ and $i(c) = -\sum_{i\in\bbU}\delta_iw'(c^{\T}\bx_i)\bx_i\bx^{\T}_i$. Appeal to the linearization formula in \cite{kim2009}, we obtain $\widehat{\theta}_{w}(\widehat{c})=N^{-1}\sum_{i\in\bbU}\psi_{i}(\delta_{i},y_{i},\bx_{i};c^{*})+o_{\P}(n^{-1/2})$ where $$\psi_{i}(\delta_{i},y_{i},\bx_{i};c^{*})	=B(c^{*})\bx_{i,\rmsc}+\delta_{i}w_{i}(c^{*})\{y_{i}-B(c^{*})\bx_{i}\},$$ 
and $B(c)=-\Gamma(c)\left\{ i(c)\right\} ^{-1}$, $\bx_{i,\rmsc}=\left\{ \bx_{1i}^{\T},\bx_{1i}^{\T}M_{\bbS}+\bx_{2i}^{\T}(I_{q}+R_{\bbS})\right\}^{\T}$. By the central limit theorem, it follows that $n^{1/2}\left\{ \widehat{\theta}_{w}(\widehat{c})-\theta_N\right\} \rightarrow N\left[0,nN^{-2}\var\left\{ \sum_{i\in\bbU}\psi(\delta_{i},y_{i},\bx_{i};c^{*})-\theta_N\mid X_{\bbU}\right\} \right]$, where $\var(\cdot)$ is taken with respect to the joint distribution of the mixed-effects model $\zeta$ and the sampling mechanism $\delta$.
\end{proof}
We now prove the double robustness of the soft-calibration estimator in a sense that $\widehat{\theta}_w$ is consistent if either the outcome $y_{ij}$ follows a linear mixed-effects model or $Q(w)$ entails a correct propensity score model for $P(\delta_i=1\mid x_i,u)$. 
\begin{assumption}[Outcome model]
\label{assmp:outcome}
The linear mixed-effects model (\ref{mixed}) is a correct specification for the study variable $y_{i}$, that is, $y_i = x_{1i}^{\T}\beta + 
x_{2i}^{\T} u + e_i$.
\end{assumption}
\begin{assumption}[Calibration weight model]
\label{assmp:weight}
 The parametric model $w(c^{\T}x)$ induced by the objective function $Q(w)$ is a correct specification for $\pi_i^{-1}={P(\delta_i=1\mid x_i,u)}^{-1}$, that is, $w(c_t^{*\T}x)=\pi_i^{-1}$, where $c_t^{*}$ is the true parameter.
\end{assumption}

\begin{proof}
1. Under Assumption \ref{assmp:outcome}, we have
\begin{align*}
    \widehat{\theta}_{w}	=	\frac{1}{N}\sum_{i\in\bbU}\delta_{i}w(\widehat{c}^\T \bx_i) e_i +\frac{1}{N}\sum_{i\in\bbU} (\bx_{1i}^{\T}\beta + \bx_{2i}^{\T}\bu)+T_{r}\bu,
\end{align*}
where $T_r = N^{-1}\sum_{i\in\bbU}(\bx_{1i}^{\T}M_{\bbS}+\bx_{2i}^{\T}R_{\bbS})$. The difference between the soft calibration $\widehat{\theta}_w$ and the finite population mean $\theta_N$ under the linear mixed model is
\begin{align}
N^{-1}\left\{\sum_{i=1}^N \delta_i w\left(\widehat{c}^{\T} x_i\right) y_i-\sum_{i=1}^N y_i\right\} &=N^{-1}\sum_{i=1}^N\left\{\delta_i w\left(\widehat{c}^{\T} x_i\right)-1\right\}\left(x_{1 i}^{\T} \beta+x_{2 i}^{\T} u+e_i\right) \nonumber\\
&=N^{-1} \gamma_n\left(1_N^{\T} X_{1,\bbU} D_{12}+1_N^{\T} X_{2,\bbU} D_{22}\right) D_q^{-1} u \label{eq:bias-soft}\\
&+N^{-1} \sum_{i=1}^N\left\{\delta_i w\left(\widehat{c}^{\T} x_i\right)-1\right\} e_i,  \label{eq:remainder-soft}
\end{align}
where (\ref{eq:bias-soft}) is incurred by the approximated imbalance of $X_2$ and (\ref{eq:remainder-soft}) is incurred by the random error. For (\ref{eq:bias-soft}), we have
\begin{align*}
\var_{\zeta}\{(1_N^{\T}X_{1,\bbU}D_{12}+1_N^{\T}X_{2,\bbU}D_{22})D_q^{-1}u\mid X_{\bbU},\bbS\}
&= 
\{(1_N^{\T}X_{1,\bbU}D_{12}+1_N^{\T}X_{2,\bbU}D_{22})D_q^{-1/2}\}^{\otimes 2}\sigma_u^2\\
&=O_{\P}(N^{2}n^{-2}q),
\end{align*}
where $\max _{i \in U}\left\|x_i\right\|^2 \leq q C$ by Assumption \ref{assm:regularity}(b). Also, note that $\|D_{12} D_{21}\|-\|D_{11}^2\| \leq\|D_{11}^2+D_{12} D_{21}\| \leq\|\{X_{\bbS}^{\T} X_{\bbS}+\gamma_n \text{diag}(0, D_q^{-1})\}^{-2}\|$ and $\|D_{11}\| \leq\|\{X_{\bbS}^{\T} X_{\bbS}+\gamma \text{diag}(0, D_q^{-1})\}^{-1}\|$ by Assumption \ref{assm:regularity}(a). Hence, $\gamma_n\|D_{12}\| \leq \sqrt{2} \gamma_n\|\{X_{\bbS}^{\T} X_{\bbS}+\gamma_n \text{diag}(0, D_q^{-1})\}^{-1}\| \leq \sqrt{2} C_0 \gamma_n n^{-1}$, $\gamma_n\|D_{22}\| \leq C_0 \gamma_n n^{-1}$ for constant $C_0$ \citep[Lemma 1]{dai2018broken}. For (\ref{eq:remainder-soft}), we have
\begin{align*}
\var_{\zeta}\left[
\sum_{i=1}^N
\{\delta_iw(\widehat{c}^{\T}x_i)-1\}e_i\mid X_{\bbU},\bbS
\right]&=\sum_{i=1}^N
\{\delta_iw(\widehat{c}^{\T}x_i)-1\}^2\sigma_e^2=O_{\P}(N^2n^{-1}).
\end{align*}
Therefore, we have
$$
\E\{
(\widehat{\theta}_w-\theta_N)^2
\mid X_{\bbU}
\}=
\E_{\delta}
\left[
\E_{\zeta}\{
(\widehat{\theta}_w-\theta_N)^2
\mid X_{\bbU},\bbS
\}\mid X_{\bbU}\right]=
O_{\P}(\gamma_n^2n^{-2}q)+
O_{\P}(n^{-1}),
$$
which proves the consistency of $\widehat{\theta}_{w}$ when $\gamma_n=o_{\P}(n^{1/2}q^{-1/2})$.

2. Under Assumption \ref{assmp:weight}, we take the Taylor series of $\widehat{\theta}_w-\theta_N$ around the true parameter $c_t^*$:
$$
\widehat{\theta}_w-\theta_N = 
N^{-1}\sum_{i\in\bbU}\{\delta_iw(c_t^{*\T}x_i)-1\}y_i + 
N^{-1}\sum_{i\in\bbU}\delta_i\{w(\widehat{c}^{\T}x_i)-w(c_t^{*\T}x_i)\}y_i,
$$
where the second term is $O_{\P}=(n^{-1/2})$ under Lemma \ref{lm:c_true}. For the first term, we know that $\E_{\delta}\left[\{\delta_iw(c_t^{*\T}x_i)-1\}y_i\mid X_{\bbU},u, Y_{\bbS}\right]=0$ under Assumption \ref{assum:laten_ignore}. By Chebyshev's inequality, it suffices to show that $N^{-2}\var\left[\sum_{i\in\bbU}\{\delta_iw(c_t^{*\T}x_i)-1\}y_i\mid X_{\bbU},u\right]$ converges to zero, that is,
\begin{align*}
    \var&\left[\sum_{i\in\bbU}\{\delta_iw(c_t^{*\T}x_i)-1\}y_i\mid X_{\bbU},u\right] =
    \var_{e}\left(\E_{\delta}\left[\sum_{i\in\bbU}\{\delta_iw(c_t^{*\T}x_i)-1\}y_i\mid X_{\bbU},u,Y_{\bbS}\right]\mid X_{\bbU},u\right)\\
    &+ \E_{e}\left(\var_{\delta}\left[\sum_{i\in\bbU}\{\delta_iw(c_t^{*\T}x_i)-1\}y_i\mid X_{\bbU},u,Y_{\bbS}\right]\mid X_{\bbU},u\right)\\
    &=\E_{e}\left\{\sum_{i\in\bbU}\pi_{i}^{-1}(1-\pi_i)y_i^2\mid X_{\bbU},u\right\},
\end{align*}
which is in the order of $O_{\P}(N^2n^{-1})$ under Assumption \ref{assum:positivity}. Therefore, we have
\begin{align*}
    \E\{
(\widehat{\theta}_w-\theta_N)^2
\mid X_{\bbU}
\}&=
\E_{u}
\left[\E\{
(\widehat{\theta}_w-\theta_N)^2
\mid X_{\bbU},u
\}\mid X_{\bbU}\right]\\
&=\E_{u}
\left[\E\{
(\widehat{\theta}_w-\theta_N)
\mid X_{\bbU},u
\}^2 + 
\var\{
(\widehat{\theta}_w-\theta_N)
\mid X_{\bbU},u
\}
\mid X_{\bbU}\right]=O_{\P}(n^{-1}),
\end{align*}
which completes the proof.
\end{proof}

\subsection{Proof of Corollary \ref{coro_theta_cluster}}\label{sec:proof of corollary_1}
Assume the following regularity conditions hold in the presence of cluster-specific nonignorable missingness.

\begin{assumption}[Regularity conditions] \label{assum:regular-sample}
(a) The sequence of finite populations $\mathcal{F}_{N}$ is a random
sample from a super-population where $N^{-1}\sum_{i=1}^K\sum_{j=1}^{N_i}y_{ij}^{2+\alpha}$ is bounded for some $\alpha>0$. And the design-weighted complete-data estimator
$\widehat{\theta}_{n}=N^{-1}\sum_{i=1}^{k}\sum_{j=1}^{n_{i}}d_{ij}y_{ij}$
with $d_{ij}=d_{i}N_{i}/n_{i}$ satisfying that $\text{var}(\widehat{\theta}_{n})=O(n^{-1})$
and $\text{var}(\widehat{\theta}_{n})^{-1/2}(\widehat{\theta}_{n}-\theta_N)\rightarrow N(0,1)$; (b) The sampling weights satisfy $\underline{d}<d_{ij}n/N<\overline{d}$
for $i=1,\ldots,k$ and $j=1,\ldots,n_{i}$ for some constants $\underline{d}$
and $\overline{d}>0$; (c) The propensity score is uniformly bounded, that is, $\underline{e}<P(\delta_{ij}=1\mid\bx_{ij},a_{i})<\overline{e}$
for all $\bx_{ij},a_{i}$ and for some constants $0<\underline{e}$, $\overline{e}<1$. 
\end{assumption}

Assumption \ref{assum:regular-sample}(a) is standard in the survey
sampling literature to derive the asymptotic properties of the design-weighted
estimator \citep{fuller09}. Assumptions \ref{assum:regular-sample}(b)
and (c) are the counterparts of Assumption \ref{assum:positivity}, which are required for identification. Follow the similar arguments in the proof of Theorem \ref{thm:theta_linear}, the first-order Taylor expansion of $\widehat{\theta}_{w}$ is 
$$\begin{aligned} 
& \widehat{\theta}_{w}(\widehat{c})=\widehat{\theta}_{w}(c^*)\\
 & -N^{-1}\left\{ \sum_{i=1}^k\sum_{j=1}^{n_i} d_{ij}\delta_{ij} y_{ij}
 \frac{\partial w(c_0^{*\T}x_{ij}+c_1^{*\T}z_{ij})}{\partial c^{\T}} \right\}
 \\
 &\times \left\{ \sum_{i=1}^k\sum_{j=1}^{n_i} d_{ij}\delta_{ij}\frac{\partial w(c_0^{*\T}x_{ij}+c_1^{*\T}z_{ij})}{\partial c^{\T}}
 \begin{pmatrix}\bx_{ij}\\
z_{ij}
\end{pmatrix}^{\T}\right\} ^{-1}\frac{\partial G(c^*)}{\partial c}+o_{\P}(n^{-1/2}).
\end{aligned}$$
Let 
\begin{align*}
    &\Gamma(c)=\sum_{i=1}^k\sum_{j=1}^{n_i} d_{ij}\delta_{ij}\frac{\partial w(c_0^{\T}x_{ij}+c_1^{\T}z_{ij})y_{ij}}{\partial c^{\T}},\quad i(c)=\sum_{i=1}^k\sum_{j=1}^{n_i} d_{ij}\delta_{ij}\frac{\partial w(c_0^{\T}x_{ij}+c_1^{\T}z_{ij})}{\partial c^{\T}}
\begin{pmatrix}\bx_{ij}\\
z_{ij}
\end{pmatrix}^{\T},
\end{align*}
and 
$$
\frac{\partial G(c)}{\partial c}=U(c)=\left[\begin{array}{c}
\sum_{i=1}^{k}\sum_{j=1}^{n_{i}}d_{ij}(\delta_{ij}w_{ij}-1)\bx_{ij}\\
\sum_{i=1}^{k}\sum_{j=1}^{n_{i}}d_{ij}\left\{ \delta_{ij}w_{ij}-\left(I_{k}+R_{\bbS}^{\T}\right)\right\} z_{ij}-\sum_{i=1}^{k}\sum_{j=1}^{n_{i}}d_{ij}M_{\bbS}^{\T}\bx_{ij}
\end{array}\right].
$$
Rearranging the terms leads to
$$
\begin{aligned} & n^{1/2}\left\{ \widehat{\theta}_{w}(\widehat{c})-\theta_N\right\} =n^{1/2}\left\{ \widehat{\theta}_{w}(c^*)-\theta_N\right\} +n^{1/2}N^{-1}\Gamma(c^*)i^{-1}(c^*)U(c^*)+o_{\P}(1)\\
 & =n^{1/2}\left[\frac{1}{N}\sum_{i=1}^{k}d_{i}\frac{N_{i}}{n_{i}}\sum_{j=1}^{n_{i}}\left\{ B(c^{*})\bx_{ij,\rmsc}+\delta_{ij}w(c_0^{*\T}x_{ij}+c_1^{*\T}z_{ij})\eta_{ij}(c^{*})\right\} -\theta_N\right]+o_{\P}(1)\\
 & =n^{1/2}\left\{ \frac{1}{N}\sum_{i=1}^{K}d_{i}\psi_{i}(c^{*})-\theta_N\right\} +o_{\P}(1),
\end{aligned}
$$
where 
$$
\begin{aligned}
&\psi_{i}(c^{*})	=\frac{N_{i}}{n_{i}}\sum_{j=1}^{n_{i}}\left\{ B(c^{*})\bx_{ij,\rmsc}+\delta_{ij}w(c_0^{*\T}x_{ij}+c_1^{*\T}z_{ij})\eta_{ij}(c^{*})\right\} ,\\
&\bx_{ij,\rmsc}	=\left\{ \begin{array}{c}
\bx_{ij}\\
(I_{k}+R_{\bbS}^{\T})z_{ij}+M_{\bbS}^{\T}\bx_{ij}
\end{array}\right\} ,\quad\eta_{ij}(c^{*})=y_{ij}-B(c^{*})\begin{pmatrix}\bx_{ij}\\
z_{ij}
\end{pmatrix},
\end{aligned}
$$
and $B(c^{*})=-\Gamma(c^*)\left\{i(c^*)\right\} ^{-1}$.

\subsection{Soft Calibration under Cluster-specific Nonignorable Treatment Mechanism}\label{sec:nonignore_trt}

Causal inference with clustered data has attracted increasing attention recently \citep{li2013propensity, leite2015evaluation, xiang2015propensity,schuler2016propensity}. Following the potential outcomes framework \citep{rubin1974estimating}, consider the finite population $\mathcal{F}_{N}=\{A_{ij},\bx_{ij},y_{ij}(0),y_{ij}(1):i=1,\ldots,K,j=1,\ldots,N_{i}\}$, where $i$ indexes the cluster and $j$ indexes the unit within each cluster, $A_{ij}\in\{0,1\}$ is the binary treatment assignment, $\bx_{ij}\in\mathbb{R}^{p}$ is a vector of pre-treatment covariates, and $y_{ij}(a)$ is the potential outcome that would be observed had unit $j$ in cluster $i$ received treatment $a$, and $N=\sum_{i=1}^{k}N_{i}$ is the population size. Assume that the sampled cluster data $\{A_{ij},\bx_{ij},y_{ij}(0),y_{ij}(1):j=1,\ldots,n_{i}\}$ follows a superpopulation model for $i=1,\ldots,k$. Specifically, the potential outcome follows a linear mixed-effects model
\begin{equation*}
    y_{ij}(a)=\bx_{ij}^{\T}\beta_{a}+z_{ij}^{\T}\bu_{a}+e_{ij}(a),
\end{equation*}
where $\bx_{ij}$ and $z_{ij}$ adopt the same definition in Corollary \ref{coro_theta_cluster}, $\beta_{a}=\beta_{1}1(a=1)+\beta_{0}1(a=0)$ and $u_{a}=u_{1}1(a=1)+u_{0}1(a=0)$ represent fixed effects and latent cluster-level confounding effects for treatment $a$, respectively, and $e_{ij}(a)$ are independent errors for $a=0,1$. Based on potential outcomes, the causal estimand is the population average treatment effect $\tau_N=N^{-1}\sum_{i=1}^{K}\sum_{j=1}^{N_{i}}\{y_{ij}(1)-y_{ij}(0)\}$. 

We consider the two-stage cluster sampling design in $\S$\ref{subsec:Cluster-specific-nonignorable-mi} for data collection, where the observed outcome is $y_{ij}=y_{ij}(A_{ij})$. Our notation implicitly makes the Stable Unit and Treatment Version assumption, where the potential outcomes for each unit are not affected by the treatments assigned to other units \citep{rubin1978bayesian}.
\begin{assumption}[Regularity conditions]\label{assum:regular-causal}
 Assumption \ref{assum:regular-sample}(a) and (b) hold. Assume further that the propensity score is uniformly bounded, that is, $\underline{e}<P(A_{ij}=a\mid {x}_{ij},u_{a})<\overline{e}$, $a=0,1$ for any $\bx_{ij}, u_{a}$ and some constants $0<\underline{e},\overline{e}<1$. 
\end{assumption}
Similarly, we impose the following soft calibration constraints to balance covariates and random cluster effects between treatment $a=0,1$ and their combined group: 	
\begin{align*}
  &\sum_{i=1}^{k}\sum_{j=1}^{n_{i}}d_{ij}A_{ij}w_{1,ij}\bx_{ij}=\sum_{i=1}^{k}\sum_{j=1}^{n_{i}}d_{ij}(1-A_{ij})w_{0,ij}\bx_{ij}=\sum_{i=1}^{k}\sum_{j=1}^{n_{i}}d_{ij}\bx_{ij},\\
	&\sum_{i=1}^{k}\sum_{j=1}^{n_{i}}d_{ij}A_{ij}w_{1,ij}z_{ij}=\sum_{i=1}^{k}\sum_{j=1}^{n_{i}}d_{ij}z_{ij}+\sum_{i=1}^{k}\sum_{j=1}^{n_{i}}d_{ij}M_{\bbS,1}^{\T}\bx_{ij}+\sum_{i=1}^{k}\sum_{j=1}^{n_{i}}d_{ij}R_{\bbS,1}^{\T}z_{ij},\\
	&\sum_{i=1}^{k}\sum_{j=1}^{n_{i}}d_{ij}(1-A_{ij})w_{0,ij}z_{ij}=\sum_{i=1}^{k}\sum_{j=1}^{n_{i}}d_{ij}z_{ij}+\sum_{i=1}^{k}\sum_{j=1}^{n_{i}}d_{ij}M_{\bbS,0}^{\T}\bx_{ij}+\sum_{i=1}^{k}\sum_{j=1}^{n_{i}}d_{ij}R_{\bbS,0}^{\T}z_{ij},
\end{align*}
where $w_{a,ij}$ is the calibration weight for treatment $a=0, 1$. By minimizing the chosen loss functions for the treated $(A_{ij}=1)$ and the control $(A_{ij}=0)$ as our objectives, our soft calibrated estimator of $\tau_N$ can be obtained by
$$
\widehat{\tau}_{w}=\frac{1}{N}\sum_{i=1}^{k}\sum_{j=1}^{n_{i}}d_{ij}\left\{ A_{ij}w(c_{10}^{*\T}x_{ij} + c_{11}^{*\T}z_{ij})y_{ij}-(1-A_{ij})w(c_{00}^{*\T}x_{ij} + c_{01}^{*\T}z_{ij})y_{ij}\right\},$$
$c_{0}= (c_{00}^{\T}, c_{01}^{\T})^{\T}$, and $c_{1}= (c_{10}^{\T}, c_{11}^{\T})^{\T}$. Following Theorem \ref{thm:theta_linear}, we derive the asymptotic properties of $\widehat{\tau}_{w}$.

\begin{corollary}\label{coro_cluster_soft_causal}
Under Assumption \ref{assum:laten_ignore}, \ref{assm:regularity}, and \ref{assum:regular-causal}, if either the outcome $y_{ij}$ follows a linear mixed-effects model or $Q(w)$ entails a correct propensity score model, we have 
$n^{1/2}(\widehat{\tau}_{w}-\tau_N)\rightarrow N(0,V_{1})$ as $n\rightarrow\infty$, where $V_{1}=\lim_{n\rightarrow\infty}nN^{-2}\var\left\{ \sum_{i=1}^{k}d_{i}\phi_{i}(c_{0}^{*},c_{1}^{*})-\tau_N \mid X_{\bbU}\right\}$  with
\begin{align*}
    \phi_i(c_0^*,c_1^*) = 
    \frac{N_i}{n_i}\sum_{j=1}^{n_i}
    &\left[
    \{A_{ij}w(c_{10}^{*\T}x_{ij} + c_{11}^{*\T}z_{ij})
    \{y_{ij}- B_1(c_1^*)(x_{ij}^{\T},z_{ij}^{\T})^{\T}\} + B_1(c_1^*)x_{ij,\rmsc}\}\right.\\
    &\left.-
    \{
    (1-A_{ij})w(c_{00}^{*\T}x_{ij} + c_{01}^{*\T}z_{ij})  \{y_{ij}- B_0(c_0^*)(x_{ij}^{\T},z_{ij}^{\T})^{\T}\} + B_0(c_0^*)x_{ij,\rmsc}
    \}
    \right],
    \end{align*}
    where $\{B_0(c),B_1(c)\}$ are defined similarly as in Corollary \ref{coro_theta_cluster} with $\delta_{ij}$ replaced by $1-A_{ij}$ and $A_{ij}$. 
\end{corollary}

\subsection{Detailed Implementation for Soft Calibration}\label{sec:implement_details}
\begin{proof}[Poof of Proposition \ref{prop:dual-soft}]
    By soft calibration, our goal is to minimize $\sum_{i\in \bbU}\delta_i Q(w_i)$ while subject to $\sum_{i\in \bbU}\delta_i w_i x_{i} = NT_X$. We can show its dual problems by assigning a non-negative Lagrange multiplier $c$:
    \begin{align*}
        &\min_w\left\{
    \sum_{i\in\bbU} \delta_i Q(w_i) -
    \langle c, \sum_{i\in\bbU}\delta_i w_i x_i - NT_X\rangle
    \right\}
    = \sum_{i\in\bbU} \delta_i
    \min_{w_i}
    \left\{
    Q(w_i) -
    \langle c, w_i x_i\rangle
    \right\} + \langle c, NT_X \rangle
    \\
    &= \sum_{i\in\bbU} \delta_i
    \min_{w_i}
    \left\{
    Q(w_i) -
    \langle c^\T x_i, w_i\rangle
    \right\} + \langle c, NT_X \rangle
    \\
    &= -\sum_{i\in\bbU} \delta_i
    g(c^\T x_i) + \langle c, NT_X \rangle,
    \end{align*}
    where $g(z)$ is the convex conjugate function of $Q(w)$, defined by
    $$
    g(z) = \sup_{w}\{\langle z, w\rangle - Q(w)\} = 
    z\cdot(Q')^{-1}(z) - Q\{(Q')^{-1}(z)\},
    $$
    with $w(z) = g'(z) = (Q')^{-1}(z)$. By \cite{tseng1987relaxation}, the dual problem will be the unconstrained convex optimization 
    $$
    \widehat{c} = \arg\max_c\left\{
    -\sum_{i\in\bbU} \delta_i
    g(c^\T x_i) + N\langle c, T_X \rangle
    \right\}=
    \arg\min_c\left\{ \sum_{i\in\bbU} \delta_i
    g(c^\T x_i) - N\langle c, T_X \rangle\right\},
    $$
    which is proved to be convex in Lemma \ref{lem:unconstrained_convex} and can be solved via Newton-type method in Algorithm \ref{alg:newton-minimization}; see other methods for optimization within a bounded support $\mathbb{W}$ in \cite{devaud2019deville}.
    \end{proof}
    \begin{lemma}
\label{lem:unconstrained_convex}
$g(z)$ is a convex function as long as $Q(w)$ is convex.
\end{lemma}
    \begin{algorithm}[!ht]
\caption{\label{alg:newton-minimization} Newton-type algorithm for solving the unconstrained convex problem $L(c)$}
\textbf{Input}: $c^{(0)}=0$, $w^{(0)}=w(0)$, error tolerance $tol$, and the number of iterations $B$.\\
\For {$b=1,\cdots,B$}
{
$c^{(b)} = c^{(b-1)} - \{\nabla^2 L(c^{(b-1)})\}^{-1}\nabla L(c^{(b-1)})$.\\
$w^{(b)} = w\{c^{(b)\T} x_i\},i\in\bbS$.\\
\If{$\|w^{(b)} - w^{(b-1)}\|_{\infty}<tol$}{\text{break}}
}
$\widehat{w}=w\{c^{(B)\T} x_i\},i\in\bbS$.
\end{algorithm}
Let 
$$
L(c) = \sum_{i\in\bbU} \delta_i
g(c^\T x_i) - N\langle c, T_X \rangle,
\quad \nabla L(c)=\sum_{i\in \bbU}\delta_i g'(c^{\T}x_i)x_i,
\quad \nabla^2 L(c)=\sum_{i\in \bbU}\delta_i g''(c^{\T}x_i)x_ix_i^{\T},
$$
therefore, if one post-stratum in the population level is not selected in the sample, it will not contribute to the gradient. The Moore-Penrose generalized inverse can be used if $\nabla^2 L(c)$ is not invertible. Also, a small value $\epsilon>0$ can be added to the diagonal of $\nabla^2 L(c)$ to bypass the non-invertibility; other approaches for relaxing or prioritizing the restrictions can be found in \cite{montanari2007multiple} and \cite{williams2019optimization}. 

\subsection{Detailed Discussion for Bias-corrected Estimator}\label{sec:proof_rmk_dr}
Suppose that the outcome model $\widehat{\mu}_i$ is parameterized as $\widehat{\mu}=\mu(x_i;\widehat{\xi})$, where $\widehat{\xi}$ is the solution to $\Phi(\xi)=\sum_{i\in\bbU}\Phi(x_i,y_i,\delta_i;\xi)=0$. Let $\xi^*$ be the solution to $\E\{\Phi(x,y,\delta;\xi)\mid x_i\}=0$. Assume the following regularity conditions hold.
\begin{assumption}[Regularity conditions for $\Phi(x,y,\delta;\xi)$]
\label{assum:regularity-Phi-xi}
    (a) There exists an open set $\Xi$ that contains the true parameter $\xi^*$, where the first and the second derivative of $\Phi(x,y,\delta;\xi)$ exist. (b) For any $\xi\in\Xi$, the first derivative of $\Phi(x,y,\delta;\xi)$ is finite. Also, for any $\xi\in\Xi$, there exist function $B(x_i,y_i,\delta_i)$ such that $|\partial^2\Phi(x,y,\delta;\xi)/{\partial \xi_j\partial \xi_k^{\T}}|<B(x_i,y_i,\delta_i)$ for $j,k\in\{1,\cdots,p+q\}$, where $\E\{B(x_i,y_i,\delta_i)\mid x_i\}<\infty$.
\end{assumption}
Under Assumption \ref{assum:regularity-Phi-xi}, the asymptotic properties of $\widehat{\theta}_{\dr}=\widehat{\theta}_{\dr}(\widehat{c},\widehat{\xi})$ can be derived by Taylor's Theorem
\begin{align*}
   \widehat{\theta}_{\dr}(\widehat{c},\widehat{\xi})&=
   {\theta}_{\dr}({c}^*,{\xi}^*)+\frac{\partial {\theta}_{\dr}({c}^*,{\xi}^*)}{\partial ({c}^*,{\xi}^*)^{\T}}
   \begin{pmatrix}
       \widehat{c} - {c}^*\\
       \widehat{\xi} - {\xi}^*
   \end{pmatrix}+o_{\P}(n^{-1/2})\\
   &={\theta}_{\dr}({c}^*,{\xi}^*)+
   \begin{pmatrix}
       N^{-1}\Gamma_{c}({c}^*,{\xi}^*)\\
       N^{-1}\Gamma_{\xi}({c}^*,{\xi}^*)
   \end{pmatrix}^{\T}
   \begin{pmatrix}
   \{i_c(c^*)\}^{-1}U(c^{*})\\
    \{i_\xi(\xi^*)\}^{-1}\Phi(\xi^{*})
   \end{pmatrix}+o_{\P}(n^{-1/2})\\
   &=\frac{1}{N} \sum_{i\in\bbU}\psi_{i,\dr}(c^*,\xi^*)+o_{\P}(n^{-1/2}),
\end{align*}
where 
\begin{align*}
    &\Gamma_{c}(c,\xi)=\sum_{i\in\bbU}\delta_i w'(c^{\T}x_i)x_i\{y_i-\mu(x_i;\xi)\},\quad 
    \Gamma_{\xi}(c,\xi)=-\sum_{i\in\bbU}\{\delta_iw(c^{\T}x_i)-1\}\cdot
       \partial\mu_i/\partial\xi,\\
       &i_c(c) = -\sum_{i\in\bbU}\delta_i w''(c^{\T}x_i)x_ix_i^{\T},
       \quad i_\xi(\xi)=-\sum_{i\in\bbU}\partial\Phi(x_i,y_i,\delta_i;\xi)/\partial \xi,\\
       &\psi_{i,\dr}(c^*,\xi^*)=\delta_i w(c^{*\T}x_i)\{y_i-\mu(x_i;\xi^*)-B(c^*,\xi^*)x_i\} + 
        B(c^*,\xi^*)x_{i,\rmsc}+\mu(x_i;\xi^*)\\
       &B(c^*,\xi^*)=-\Gamma_c(c^*,\xi^*){i_c(c^*)}^{-1},\quad \bx_{i,\rmsc}=\left\{ \bx_{1i}^{\T},\bx_{1i}^{\T}M_{\bbS}+\bx_{2i}^{\T}(I_{q}+R_{\bbS})\right\}^{\T}.
\end{align*}
Hence, we can show that $\widehat{\theta}_{\dr}$ is doubly robust in a sense if the outcome $y_{i}$ is correctly specified by $\mu(x_i;\xi)$ or the propensity score $\pi_i$ is correctly specified by $w(c^{\T}x_i)^{-1}$ following the arguments in \cite{yang2019combining}. In addition, the estimator $\widehat{\theta}_{\dr}$ is likely to have a smaller variance than $\widehat{\theta}_w$ provided that the residuals $y_i- \mu(x_i;\xi)- B(c^{*},\xi^*)x_i$ have smaller variation
than $y_i-B(c^{*})x_i$ as long as $\mu(x_i;\xi)$ is able to partially explain the post-calibration bias. 

More importantly, the asymptotic distribution of $\widehat{\theta}_{\dr}$ can be obtained via the joint randomization framework in a similar manner, which accounts for the variations from the super-population model $\zeta$ and the sampling mechanism $\delta$. As a result, we have $n^{1/2}(\widehat{\theta}_{\dr}-\theta_N)\rightarrow N(0,V_1+V_2)$, where
 $$
V_1=\lim _{n \rightarrow \infty} \frac{n}{N^2} E_\zeta\left[\var_\delta\left\{\sum_{i \in \bbU} \psi_{i,\dr}(c^*,\xi^*) -\theta_N\mid X_{\bbU}, u, Y_{\bbS}\right\}\mid X_{\bbU}\right],
$$
and
$$
V_2=\lim _{n \rightarrow \infty} \frac{n}{N^2} \var_\zeta\left[E_\delta\left\{\sum_{i \in \bbU} \psi_{i,\dr}(c^*,\xi^*) -\theta_N\mid X_{\bbU}, u, Y_{\bbS}\right\}\mid X_{\bbU}\right] .
$$
\begin{example}
To raise an example for the bias-corrected estimator, we fit the outcome model by best linear unbiased predictor as $\widehat{\mu}_i = x_{1i}^{\T}\widehat{\beta} + x_{2i}^{\T}\widehat{u}$ and we can show that $\widehat{\theta}_{\dr} = \widehat{\theta}_w - N^{-1}\sum_{i\in\bbU}\{\delta_i w(\widehat{c}^{\T}x_i)-1\}\widehat{\mu}_i$ is consistent under either Assumption \ref{assmp:outcome} or \ref{assmp:weight}.

1. Under Assumption \ref{assmp:outcome}, we have the bias-corrected estimator as:
\begin{align*}
    \widehat{\theta}_{\dr} &= 
N^{-1}\sum_{i\in\bbU}\delta_i w(\widehat{c}^{\T}x_i)y_i
-N^{-1}\sum_{i\in\bbU}\{\delta_i w(\widehat{c}^{\T}x_i)-1\}(x_{1i}^{\T}\widehat{\beta} + x_{2i}^{\T}\widehat{u}).
\end{align*}
By subtracting it from the finite-population mean, the difference will be 
\begin{align*}
    \widehat{\theta}_{\dr} -\theta_N  &= 
        \frac{1}{N}
        \sum_{i\in\bbU} 
        \{\delta_i w(\widehat{c}^{\T}x_i) -1\}(y_i - \widehat{\mu}_i) \\
        &=
        N^{-1}
        \gamma_n
        (1_N^{\T}X_{1,\bbU}D_{12}+ 1_N^{\T}X_{2,\bbU}D_{22})D_q^{-1}(u-\widehat{u})+
        N^{-1}\sum_{i=1}^N \{\delta_i w(\widehat{c}^{\T}x_i) -1\}e_i.
\end{align*}
Given the block matrix computation, we have
$$
\left(\begin{array}{c}
{\widehat{\beta}}\\
\widehat{{u}}
\end{array}\right)=\left(\begin{array}{cc}
 {D}_{11} & {D}_{12}\\
 {D}_{21} & {D}_{22}
\end{array}\right)\left(\begin{array}{c}
 {X}_{1,\bbS}^{ \T} {Y}_{\bbS}\\
 {X}_{2,\bbS}^{ \T} {Y}_{\bbS}
\end{array}\right)=\left\{ \begin{array}{c}
( {D}_{11} {X}_{1,\bbS}^{ \T}+ {D}_{12} {X}_{2,\bbS}^{ \T}) {Y}_{\bbS}\\
( {D}_{21} {X}_{1,\bbS}^{ \T}+ {D}_{22} {X}_{2,\bbS}^{ \T}) {Y}_{\bbS}
\end{array}\right\} =\left(\begin{array}{c}
 {E}_{1}^{ \T}\\
 {E}_{2}^{ \T}
\end{array}\right) {Y}_{\bbS}.$$
As pointed out in \cite{henderson1975best}, the following identity holds
\begin{align*}
    &\left(\begin{array}{c}
 {E}_{1}^{ \T}\\
 {E}_{2}^{ \T}
\end{array}\right)\left(\begin{array}{cc}
 {X}_{1,\bbS} &  {X}_{2,\bbS}\end{array}\right)	=	\left(\begin{array}{cc}
 {D}_{11} &  {D}_{12}\\
 {D}_{21} &  {D}_{22}
\end{array}\right)\left(\begin{array}{cc}
 {X}_{1,\bbS}^{ \T} {X}_{1,\bbS} &  {X}_{1,\bbS}^{ \T} {X}_{2,\bbS}\\
 {X}_{2,\bbS}^{ \T} {X}_{1,\bbS} &  {X}_{2,\bbS}^{ \T} {X}_{2,\bbS}
\end{array}\right)\\
	&=	\left(\begin{array}{cc}
 {D}_{11} &  {D}_{12}\\
 {D}_{21} &  {D}_{22}
\end{array}\right)\left\{ \left(\begin{array}{cc}
 {X}_{1,\bbS}^{ \T} {X}_{1,\bbS} &  {X}_{1,\bbS}^{ \T} {X}_{2,\bbS}\\
 {X}_{2,\bbS}^{ \T} {X}_{1,\bbS} &  {X}_{2,\bbS}^{\T} {X}_{2,\bbS}+\gamma_n D_{q}^{-1}
\end{array}\right)\right.
		\left.-\left(\begin{array}{cc}
 {0} &  {0}\\
 {0} & \gamma_n D_{q}^{-1}
\end{array}\right)\right\} \\
	&=	\left(\begin{array}{cc}
I_{p} &  {0}\\
 {0} & I_{q}
\end{array}\right)-\left(\begin{array}{cc}
 {0} & \gamma_n {D}_{12}D_{q}^{-1}\\
 {0} & \gamma_n {D}_{22}D_{q}^{-1}
\end{array}\right),
\end{align*}
which means that  
\begin{align*}
   &{E}_{1}^{ \T} {X}_{1,\bbS}	=I_{p},\quad {E}_{1}^{ \T} {X}_{2,\bbS}=-\gamma_n {D}_{12}D_{q}^{-1},\\
 &{E}_{2}^{ \T} {X}_{1,\bbS}	= {0},\quad {E}_{2}^{ \T} {X}_{2,\bbS}=I_{q}-\gamma_n {D}_{22}D_{q}^{-1}. 
\end{align*}
Using these results gives us
\begin{equation}
    \begin{split}
        &\E_{e}(\widehat{\beta}\mid X,u) =  E_1^{\T}X_{1,\bbS}\beta+E_1^{\T}X_{2,\bbS}u=
    \beta -\gamma_n D_{12}D_q^{-1}u,\\
    &\E_{e}(\widehat{u}\mid X,u)=
    E_2^{\T}X_{1,\bbS}\beta+E_2^{\T}X_{2,\bbS}u
    =
    (I_q-\gamma_n D_{22}D_q^{-1})u,
    \\
    &\var_e(\widehat{\beta}	\mid {X}, {u})= 
    \sigma_e^2 E_1^{\T}E_1 = 
    \sigma_e^2 E_1^{\T}X_{1,\bbS}D_{11} +
    \sigma_e^2 E_1^{\T}X_{2,\bbS}D_{21}=
    (D_{11} -\gamma D_{12}D_{q}^{-1}D_{21})\sigma_e^2,
    \\
    &\var_{e}(\widehat{ {u}}	\mid {X}, {u})=\sigma_{e}^{2} {E}_{2}^{ \T} {E}_{2}
	=\sigma_{e}^{2} {E}_{2}^{ \T} {X}_{1,\bbS} {D}_{12}+\sigma_{e}^{2} {E}_{2}^{ \T} {X}_{2,\bbS} {D}_{22}
	=(I_{q}-\gamma_n {D}_{22}D_{q}^{-1}) {D}_{22}\sigma_{e}^{2},
    \end{split}
    \label{eq:beta_u_hat}
\end{equation}
which implies that $\widehat{u} = u +O{\P}(n^{-1/2})$. Thus, we have
\begin{align}
  \E_{\zeta}\{(\widehat{\theta}_{\dr}-\theta_N)^2\mid X_{\bbU}, \bbS\}&=
N^{-2}
\gamma_n^2\{(1_N^{\T}X_{1,\bbU}D_{12}+
1_N^{\T}X_{2,\bbU}D_{22})D_q^{-1}D_{22}D_q^{1/2}\}^{\otimes 2}\sigma_u^2\label{eq:mse_bc1}\\
&+N^{-2}\gamma_n^2\{(1_N^{\T}X_{1,\bbU}D_{12}+
1_N^{\T}X_{2,\bbU}D_{22})D_q^{-1}E_2^{\T}\}^{\otimes 2}\sigma_e^2\label{eq:mse_bc2}\\
&-N^{-2}\gamma_n(1_N^{\T}X_{1,\bbU}D_{12}+
1_N^{\T}X_{2,\bbU}D_{22})D_q^{-1}E_2^{\T}(w_{\rm sc}-1_n)\sigma_e^2\label{eq:mse_bc3}\\
&+N^{-2}(w_{\rm sc}-1_n)^{\T}(w_{\rm sc}-1_n)\sigma_e^2 +
N^{-2}(N-n)\sigma_e^2.\label{eq:mse_bc4}
\end{align}
where (\ref{eq:mse_bc1}) is induced by $|\E_{\zeta}(u-\widehat{u}\mid X_{\bbU},\bbS)|^2$, (\ref{eq:mse_bc2}) is induced by $\var_{\zeta}(\widehat{u}\mid X_{\bbU},\bbS)$, (\ref{eq:mse_bc3}) is induced by $\mathrm{cov}_{\zeta}(\widehat{u}, e\mid X_{\bbU},\bbS)$, and (\ref{eq:mse_bc4}) is induced by $\var_{\zeta}(e\mid X_{\bbU},\bbS)$.

Next, we can bound each of the above terms and obtain the mean squared error averaged over the selected samples
$$
\E\{(\widehat{\theta}_{\dr}-\theta_N)^2\mid X_{\bbU}\} = 
O_{\P}(\gamma_n^4n^{-4}q^2)
+O_{\P}(\gamma_n^2n^{-3}q)
+O_{\P}(\gamma_nn^{-5/2}q^{1/2})
+O_{\P}(n^{-1}).
$$
 Thus, the bias-corrected estimator $\widehat{\theta}_{\dr}$ is root-$n$ consistent if $\gamma_n = o_{\P}(n^{3/4}q^{-1/2})$. Compared to the conditions $\gamma_n=o_{\P}(n^{1/2}q^{-1/2})$ required for $\widehat{\theta}_w$, we allow $\gamma_n$ to grow faster with $n$ at rate $3/4$ instead of $1/2$ as being requested in Theorem \ref{thm:theta_linear} when $y_{ij}$ follows a linear mixed-effects model.
 
 2. Under Assumption \ref{assmp:weight}, we have
 \begin{equation*}
     \widehat{\theta}_{\dr} -\theta_N = 
     N^{-1}\sum_{i\in\bbU}\{\delta_iw(c_t^{*\T}x_i)-1\}(y_i-\widehat{\mu}_i) +
     N^{-1}\sum_{i\in\bbU}\delta_i\{w(\widehat{c}^{\T}x_i)-w(c_t^{*\T}x_i)\}(y_i-\widehat{\mu}_i),
 \end{equation*}
 where the second term is $O_{\P}(n^{-1/2})$ following the similar arguments in the Proof of Theorem \ref{thm:theta_linear} under the condition $\gamma_n=o_{\P}(n^{1/2}q^{-1/2})$. For the first term, it suffices to show that $$N^{-2}\var\left[\sum_{i\in\bbU}\{\delta_iw(c^{*\T}x_i)-1\}(y_i-\widehat{\mu}_i)
 \mid X_{\bbU},u
 \right]\rightarrow 0,$$ where
 \begin{align*}
      \var&\left[\sum_{i\in\bbU}\{\delta_iw(c^{*\T}x_i)-1\}(y_i-\widehat{\mu}_i)
 \mid X_{\bbU},u
 \right]
 \\
 &= \var_e\left(\E_{\delta}\left[\sum_{i\in\bbU}\{\delta_iw(c^{*\T}x_i)-1\}(y_i-\widehat{\mu}_i)
 \mid X_{\bbU},u,Y_{\bbS}
 \right]\mid X_{\bbU},u\right)\\
 &+\E_e\left(\var_{\delta}\left[\sum_{i\in\bbU}\{\delta_iw(c^{*\T}x_i)-1\}(y_i-\widehat{\mu}_i)
 \mid X_{\bbU},u,Y_{\bbS}
 \right]\mid X_{\bbU},u\right)\\
 &=\E_{e}\left\{
 \sum_{i\in\bbU}\pi_{i}^{-1}(1-\pi_i)(y_i-\widehat{\mu}_i)^2\mid X_{\bbU},u
 \right\}=O_{\P}(N^{2}n^{-1}),
 \end{align*}
 which is justified by (\ref{eq:beta_u_hat}) and completes the proof.
 \end{example}

 Next, we can establish the asymptotic properties of the bias-correct estimator $\widehat{\theta}_{\dr}$ by best linear unbiased predictor in Theorem \ref{thm:theta_dr}, Corollary \ref{coro_theta_cluster_dr} and Corollary \ref{coro_theta_causal_dr} via the standard Taylor linearization, similar to Theorem \ref{thm:theta_linear}, Corollary \ref{coro_theta_cluster} and Corollary \ref{coro_cluster_soft_causal}.
\begin{theorem}\label{thm:theta_dr}
    Suppose Assumptions \ref{assum:laten_ignore}-\ref{assm:regularity} hold, the bias-corrected soft-calibration estimator $\widehat{\theta}_{\dr}$ satisfies $\widehat{\theta}_{\dr}-\theta_N = N^{-1}\sum_{i\in\bbU}\psi_{i,\dr}(c^*) -\theta_N + o_{\P}(n^{-1/2})$, where $c^*$ is the solution to $\E\{\partial G(c)/\partial c\mid X_{\bbU},u\}=0$, \begin{align*}
        \psi_{i,\dr}(c^*) &= \delta_i \{w(c^{*\T}x_i)\eta_i(c^*) +
       NT_{r}(D_{21}x_{1i}+D_{22}x_{2i})y_i\}+
        B(c^{*})x_{i,\rmsc},\quad 
        \eta_i(c^*) = y_i - B(c^*)x_i,
    \end{align*}
    $B(c^*) = \left\{
         \sum_{i\in\bbU} \delta_i w'(c^{*\T }x_i)x_iy_i
         \right\}
         \left\{
         \sum_{i\in \bbU} \delta_i
         w'(c^{*\T }x_i) x_ix_i^{\T}
         \right\}^{-1}$, and $x_{i,\rmsc} = \{x_{1i}^{\T}, x_{1i}^{\T} M_{\bbS} + x_{2i}^{\T}(I_q + R_{\bbS})\}^\T$. Under Assumption \ref{assmp:outcome} and $\gamma_n=o_{\P}(n^{3/4}q^{-1/2})$ or Assumption \ref{assmp:weight} and $\gamma_n=o_{\P}(n^{1/2}q^{-1/2})$, we have $n^{1/2}(\widehat{\theta}_{\dr} - \theta_N) \rightarrow N(0, V_1 + V_2)$, where
         $$
        V_1=\lim _{n \rightarrow \infty} \frac{n}{N^2} E_\zeta\left[\var_\delta\left\{\sum_{i \in \bbU} \psi_{i,\dr}(c^*) -\theta_N\mid X_{\bbU}, u, Y_{\bbS}\right\}\mid X_{\bbU}\right],
        $$
        and
        $$
        V_2=\lim _{n \rightarrow \infty} \frac{n}{N^2} \var_\zeta\left[E_\delta\left\{\sum_{i \in \bbU} \psi_{i,\dr}(c^{*}) -\theta_N\mid X_{\bbU}, u, Y_{\bbS}\right\}\mid X_{\bbU}\right] .
        $$
        \label{thm:dr_est}
    \end{theorem}
    The variances $V_1$ and $V_2$ can be estimated by applying the standard variance estimation formulas with $c^{*}$ replaced by $\widehat{c}$, similar to Theorem \ref{thm:var_est}. In what follows, we present two corollaries for handling cluster-specific nonignorable missingness (Corollary \ref{coro_theta_cluster_dr}) and causal inference with unmeasured cluster-level confounders (Corollary \ref{coro_theta_causal_dr}) when the fitted outcomes $\widehat{\mu}_i$ are best linear unbiased predictors.
    
    \begin{corollary}\label{coro_theta_cluster_dr}
    Suppose the assumptions in Theorem \ref{thm:theta_dr} and Assumption \ref{assum:regular-sample} hold, we have $n^{1/2}(\widehat{\theta}_{\dr}-\theta_N)\rightarrow N(0, V_1+V_2+V_3)$ as $n\rightarrow \infty$, $n/N\rightarrow f\in[0,1]$, and 
    \begin{align*}
    &V_1 = \frac{n}{N^2}\E_{\zeta}\left(\E_{\delta}\left[\var_p\left\{\sum_{i=1}^kd_i\psi_{i,\dr}(c^*)-\theta_N\mid X_{\bbU},u,Y_{\bbS},\delta\right\}\mid X_{\bbU},u,Y_{\bbS}\right]\mid X_{\bbU}\right),\\
    &V_2 = \frac{n}{N^2}\E_{\zeta}\left(\var_{\delta}\left[\E_p\left\{\sum_{i=1}^kd_i\psi_{i,\dr}(c^*)-\theta_N\mid X_{\bbU},u,Y_{\bbS},\delta\right\}\mid X_{\bbU},u,Y_{\bbS}\right]\mid X_{\bbU}\right),\\
    &V_3 = \frac{n}{N^2}\var_{\zeta}\left(\E_{\delta}\left[\E_p\left\{\sum_{i=1}^kd_i\psi_{i,\dr}(c^*)-\theta_N\mid X_{\bbU},u,Y_{\bbS},\delta\right\}\mid X_{\bbU},u,Y_{\bbS}\right]\mid X_{\bbU}\right),
    \end{align*}
    where
    \begin{align*}
    &\psi_{i,\dr}(c^*) = 
    \frac{N_i}{n_i}\sum_{j=1}^{n_i}
    \left[\delta_{ij}\{w(c_0^{*\T}x_{ij}+c_1^{*\T}z_{ij})\eta_{ij}(c^*)+NT_r(D_{21}x_{ij}+D_{22}z_{ij})y_{ij}\} + B(c^*)x_{ij,\rmsc} \right],\\
    &B(c^*) = 
    \left\{
    \sum_{i=1}^k
    \sum_{j=1}^{n_i}
    \delta_{ij}
    d_{ij}
    w'(c_0^{*\T}x_{ij}+c_1^{*\T}z_{ij}) 
    \begin{pmatrix}
    x_{ij}\\
    z_{ij}
    \end{pmatrix}y_{ij}
    \right\}\\
    &\times
    \left\{
    \sum_{i=1}^k
    \sum_{j=1}^{n_i}
    \delta_{ij}
    d_{ij}
    w'(c_0^{*\T}x_{ij}+c_1^{*\T}z_{ij}) 
    \begin{pmatrix}
    x_{ij}\\
    z_{ij}
    \end{pmatrix}
    (x_{ij}^{\T}, z_{ij}^{\T})
    \right\}^{-1},\\
    &\eta_{ij}(c^*) = y_{ij} - B(c^*)(x_{ij}^{\T},z_{ij}^{\T})^{\T},
    \quad 
    x_{ij,\rmsc} = \{x_{ij}^{\T}, x_{ij}^{\T} M_{\bbS} + z_{ij}^{\T}(I_q + R_{\bbS})\}^\T.
    \end{align*}
    \end{corollary}
\begin{corollary}\label{coro_theta_causal_dr}
    Suppose the assumptions in Theorem \ref{thm:theta_dr} and Assumption \ref{assum:regular-causal} hold, we have $n^{1/2}(\widehat{\tau}_{\dr}-\tau_N)\rightarrow N(0, V_1+V_2+V_3)$ as $n\rightarrow \infty$, $n/N\rightarrow f\in[0,1]$, and 
    \begin{align*}
        & V_1 = \frac{n}{N^2}\E_{\zeta}\left(\E_{\delta}\left[\var_p\left\{\sum_{i=1}^kd_i\phi_{i,\dr}(c_0^*,c_1^*)-\tau_N\mid X_{\bbU},u,Y_{\bbS},\delta\right\}\mid X_{\bbU},u,Y_{\bbS}\right]\mid X_{\bbU}\right),\\
        & V_2 = \frac{n}{N^2}\E_{\zeta}\left(\var_{\delta}\left[\E_p\left\{\sum_{i=1}^kd_i\phi_{i,\dr}(c_0^*,c_1^*)-\tau_N\mid X_{\bbU},u,Y_{\bbS},\delta\right\}\mid X_{\bbU},u,Y_{\bbS}\right]\mid X_{\bbU}\right),\\
        & V_3 = \frac{n}{N^2}\var_{\zeta}\left(\E_{\delta}\left[\E_p\left\{\sum_{i=1}^kd_i\phi_{i,\dr}(c_0^*,c_1^*)-\tau_N\mid X_{\bbU},u,Y_{\bbS},\delta\right\}\mid X_{\bbU},u,Y_{\bbS}\right]\mid X_{\bbU}\right),
    \end{align*}
    where
    \begin{align*}
    &\phi_{i,\dr}(c_0^*,c_1^*) = 
    \frac{N_i}{n_i}\sum_{j=1}^{n_i}
    [A_{ij}
    \{w(c_{10}^{*\T}x_{ij} + c_{11}^{*\T}z_{ij})\eta_{1,ij}(c_1^*) + 
    NT_r(D_{21}x_{ij}+D_{22}z_{ij})y_{ij}
    \}
    + B_1(c_1^*)x_{ij,\rmsc}]\\
    &-
    \frac{N_i}{n_i}\sum_{j=1}^{n_i}
    [
    (1-A_{ij})
    \{
    w(c_{00}^{*\T}x_{ij} + c_{01}^{*\T}z_{ij})\eta_{0,ij}(c_0^*) +
    NT_r(D_{21}x_{ij}+D_{22}z_{ij})y_{ij}\}+ B_0(c_0^*)x_{ij,\rmsc}
    ],
    \end{align*}
    where $\{B_0(c),B_1(c)\}$ and $\{\eta_{0,ij}(c),\eta_{1,ij}(c)\}$ are defined similarly as in Corollary \ref{coro_theta_cluster_dr} with $\delta_{ij}$ replaced by $1-A_{ij}$ and $A_{ij}$. 
    \end{corollary}

\section{Additional Simulation Results}\label{sec:additional_sims}

\subsection{Summary}\label{sec:additional_summary}
Due to the space restriction of the main paper, we include more extensive numerical results in this section. $\S$\ref{sec:sim_res_bound} presents the simulation results for our soft calibration scheme under the bounded distance function. $\S$\ref{sec:sim_res_L_2} draws a comparison with the $L_2$ penalized regression estimator. $\S$\ref{sec:sim_res_varying_lam1_lam2} presents the numerical results under varying $\lambda_1$ and $\lambda_2$. $\S$\ref{sec:sim_res_causal} presents a simulation study of causal inference under the cluster-specific nonignorable treatment mechanism. Finally, $\S$\ref{sec:sim_res_cluster} presents additional simulation results, including another cluster setup and visual illustrations of the calibration weights. 

\subsection{Simulation Results under Bounded Weight Constraints}\label{sec:sim_res_bound}
As we separate the minimization of the distance between the weights from the soft calibration conditions, our framework is directly applicable to bound the weights as long as the chosen distance metric satisfies the conditions of Remark \ref{rmk:unique}. As summarized in Chapter 12, \cite{tille2020sampling}, the restricted distance measurements can be motivated by ensuring bounded support of the calibration weights $\mathbb{W}\subset[L, U]$ given $0\leq L\leq 1\leq U$. For example, we have the following logistic distance objective:
\begin{align}
    &Q(w) =
\frac{(1-L)(U-1)}{q_i(U-L)}
\left\{
(w-L)\log\left(
\frac{w-L}{1-L}
\right) + 
(U-w)\log\left(
\frac{U-w}{U-1}
\right)
\right\}, L<w<U,\nonumber\\
& w(z) = 
\frac{L(U-1) + U(1-L)\exp\left\{
\frac{q_i(U-L)z}{(1-L)(U-1)}
\right\}}{U-1+ (1-L)\exp\left\{
\frac{q_i(U-L)z}{(1-L)(U-1)}
\right\}},-\infty < z<+\infty,\label{eq:logit_weight}
\end{align}
where $\lim_{w\rightarrow L}Q(w)=\infty$, $\lim_{w\rightarrow U}Q(w)=\infty$, and $\partial^2Q(w)/\partial w^2 = (1-L)(U-1)/\{(w-L)(U-w)\}>0$ satisfying the conditions in Remark \ref{rmk:unique}. It can be shown that when $L\rightarrow 0$ and $U\rightarrow\infty$, the specified weight function in (\ref{eq:logit_weight}) is close to the logistic regression model if the sampling fraction $nN^{-1}$ is small; Other approaches to control the dispersion of weights are achieved by truncation and the most well-known truncated distance measurement is the linear truncated objective:
\begin{align*}
    &Q(w) = \begin{cases}
    \frac{(w-1)^2}{2q_i}&\text{ if }L< w < U,\\
    \infty & \text{ else},
    \end{cases},\quad w(z) = 
    \begin{cases}
    1 + q_iz &\text{ if }\frac{L-1}{q_i} < z< \frac{U-1}{q_i},\\
    L&\text{ if }z \leq \frac{L-1}{q_i},\\
    U&\text{ if }z \geq \frac{U-1}{q_i},
    \end{cases}
\end{align*}
which clearly satisfies the conditions in Remark 1.

\begin{table}[!ht]
\caption{\label{tab:sim-survey:truncated} Bias $(\times 10^{-2})$, variance $(\times10^{-3})$, mean squared error
$(\times10^{-3})$ and coverage probability (\%) of the estimators under cluster-specific nonignorable missingness based on $500$ simulated datasets when $(k,n_i)=(30,200)$ and $(\lambda_1=0.01,\lambda_2=10)$}
\vspace{0.15cm}
\centering
{\begin{tabular}{lccccccccccc}
\toprule
&$\widehat{\theta}_{\rm hc}$ & $\widehat{\theta}_w^{\sq}$& $\widehat{\theta}_w^{\sq,[0,10]}$& $\widehat{\theta}_w^{\sq,[0,30]}$&
$\widehat{\theta}_w^{\sq,[0,50]}$&$\widehat{\theta}_w^{\sq,[0,100]}$\\
\midrule
Bias&0.40&1.05&1.37&1.13&1.11&1.08\\
VAR&4.50&1.51&1.10&1.29&1.35&1.42\\
MSE&4.52&1.62&1.29&1.41&1.47&1.54\\
\midrule
\midrule
& & $\widehat{\theta}_w^{\me}$& $\widehat{\theta}_w^{\logit,[0,10]}$& $\widehat{\theta}_w^{\logit,[0,30]}$&
$\widehat{\theta}_w^{\logit,[0,50]}$&$\widehat{\theta}_w^{\logit,[0,100]}$\\
\midrule
Bias&&0.14&1.60&1.58&1.05&0.54\\
VAR&&1.83&4.74&2.99&2.28&1.82\\
MSE&&1.83&4.99&3.24&2.39&1.85\\
\bottomrule
\end{tabular}}
\end{table}

Table \ref{tab:sim-survey:truncated} reports the numerical results of the estimators produced by the bounded version of the loss functions under the same data generation process in the main paper. It is seen that the linear truncated objective can improve the estimation efficiency by discarding outlier weights compared to its unrestricted counterpart in almost all scenarios. However, performances are sensitive to the choice of $Q(w)$, $L$, and $U$, which should be guided by domain expertise or prior knowledge of data distribution.

\subsection{Comparison with $L_2$ Penalized Regression}\label{sec:sim_res_L_2}
We conduct additional simulations to compare the proposed soft calibration method with the $L_2$ penalized propensity score weight estimation. The $L_2$ penalized propensity score weight estimation is defined as
\begin{align*}
    &L_{2}(c) = 
N^{-1}\sum_{i\in\bbU}\delta_i
g(c^\T x_i) - \langle c, N^{-1}\sum_{i\in\bbU} x_i\rangle + \lambda\|c_{(p+1):(p+q)}\|_2^2/2,\\
&\nabla L_{2}(c) = 
N^{-1}\sum_{i\in\bbU}\delta_i
g'(c^\T x_i)x_i -
N^{-1}\sum_{i\in\bbU} x_i +
\lambda\cdot (0_{p}, c_{(p+1):(p+q)}),
\end{align*}
where $g(\cdot)$ is the convex conjugate function of $Q(\cdot)$ and $\lambda$ is a tuning parameter, which can be selected by five-fold cross-validation. Setting the gradient $\nabla L_{2}(c)=0$, we have
\begin{align}
    & N^{-1}\sum_{i\in\bbU}\delta_i
g'(c^\T x_i)x_{1i} =
N^{-1}\sum_{i\in\bbU} x_{1i}\nonumber\\
& N^{-1}\sum_{i\in\bbU}\delta_i
g'(c^\T x_i)x_{2i} -
N^{-1}\sum_{i\in\bbU} x_{2i} = \lambda\cdot c_{(p+1):(p+q)}\label{eq:norm2_x2},
\end{align}
which leads to an approximate calibration on $X_2$, similar to a ridge-type regression. Denote the $L_2$ penalized propensity score weight estimator as $\widehat{\theta}_{w}^{\ell_2}=N^{-1}\sum_{i\in\bbU}g'(\widehat{c}_{l2}^{\T}x_i)y_i$, where $\widehat{c}_{l2}$ satisfies the condition $\nabla L_{2}(\widehat{c}_{l2})=0$. Next, we use the same data generation process as mentioned in the main paper to draw a comparison among $\widehat{\theta}_{\rm hc}$, $\widehat{\theta}_{\dr}$ and $\widehat{\theta}_{w}^{\ell_2}$ using the maximum entropy objective function under the linear mixed-effects model with $(\lambda_1=0.01,\lambda_2=1)$ and $(\lambda_1=0.01, \lambda_2=10)$.

\begin{figure}[!ht]
    \centering
    \includegraphics[width=.8\linewidth]{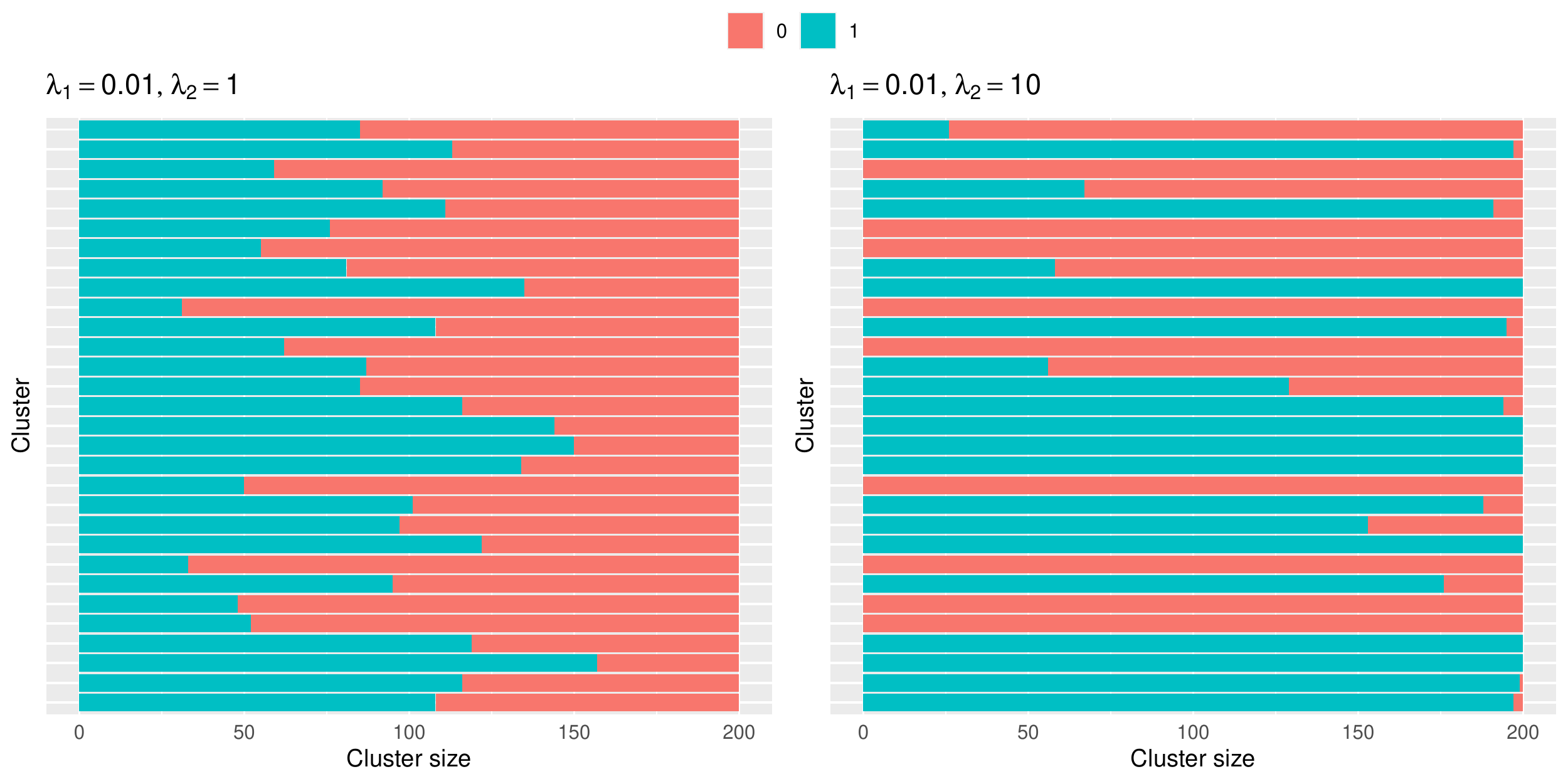}
    \caption{Cluster-specific missingness when $\lambda_1=0.01,\lambda_2=1$ (left) and $\lambda_1=0.01,\lambda_2=10$ (right); missing units are labeled by $1$ and observed units are labeled by $0$.}
    \label{fig:missing_cluster}
\end{figure}

Table \ref{tab:sim-survey:norm2} reports the simulation results based on $500$ Monte Carlo samples. We find that the $L_2$ penalized propensity score weight estimation achieves a similar performance as $\widehat{\theta}_{\rm hc}$ and $\widehat{\theta}_{\dr}$ when the between-cluster
variation of $\delta_{ij}$ is small, i.e., $\lambda_2=1$. This finding is reasonable since the data do not show stark cluster-specific missingness (Figure \ref{fig:missing_cluster} (left)), and therefore the exact calibration will not create too many extreme values. When the between-cluster variation of $\delta_{ij}$ becomes large (i.e., $\lambda_2=10$) as in Figure \ref{fig:missing_cluster} (right), the exact calibration is prone to yield extreme weights, which deteriorates the efficiency of the estimation. Both the $L_2$ penalized propensity score weight estimation and the soft calibration weighting are effective in alleviating the inflated variation induced by the abusive calibration.  However, since the approximate calibration in (\ref{eq:norm2_x2}) does not take account of the correlation structure under the linear mixed-effects model to minimize the mean square error, its numerical performance is not optimal compared to our soft calibration estimator.

\begin{table}[htbp]
\caption{\label{tab:sim-survey:norm2} Bias $(\times 10^{-2})$, variance $(\times10^{-3})$, mean squared error
$(\times10^{-3})$ and coverage probability (\%) of the estimators under cluster-specific nonignorable missingness based on $500$ simulated datasets under the linear mixed-effects model}
\vspace{0.15cm}
\centering
{\begin{tabular}{lcccccc}
\toprule
 & \multicolumn{3}{c}{$\lambda_1=0.01$, $\lambda_2=1$} & \multicolumn{3}{c}{$\lambda_1=0.01$, $\lambda_2=10$}  \\
& \multicolumn{1}{l}{bias $\times 10^{-2}$} & \multicolumn{1}{l}{var $\times 10^{-3}$} & \multicolumn{1}{l}{MSE $\times 10^{-3}$} & \multicolumn{1}{l}{bias $\times 10^{-2}$} & \multicolumn{1}{l}{var $\times 10^{-3}$} & \multicolumn{1}{l}{MSE $\times 10^{-3}$} \\
\midrule
$\widehat{\theta}_{\rm hc}$&0.10&0.78&0.78
&0.73&4.58&4.63\\
$\widehat{\theta}_{\dr}$&0.09&0.74&0.74
&0.13&1.79&1.80\\
$\widehat{\theta}_w^{\ell_2}$ 
&0.10&0.78&0.78
&2.13&1.64&2.09\\
\bottomrule
\end{tabular}}
\end{table}

\subsection{Simulation Results under Varying $\lambda_1$ and $\lambda_2$}\label{sec:sim_res_varying_lam1_lam2}

In Figure \ref{fig:varying}, we vary the parameter $\lambda_{1}$ and $\lambda_{2}$ in our simulation settings to analyze the effect of the cluster-specific variation of the outcomes and the missingness on the performance of the estimators. On the one hand, we find that $\widehat{\theta}_w^{\me}$ estimator achieves the smallest mean squared error at any $\lambda_{2}$ when $\lambda_1 = 0.01$, and this improvement becomes more evident compared to $\widehat{\theta}_{\rm hc}$ when the missing mechanism is strongly affected by the random effects, i.e., $\lambda_{2}=10$. The reason might be attributed to the fact that hard calibration creates unnecessary extreme weights and leads to inefficient estimation; see Figures \ref{fig:missing} for the density plots of their calibration weights. On the other hand, when the cluster-specific variation of $y_{ij}$ increases, i.e., $\lambda_1$ increases, it becomes more advantageous to use the bias-corrected soft-calibration estimator aided by best linear unbiased predictors, as it is more robust to the magnitude of the random effect $u$. More importantly, when the random effects have evidently strong impacts on outcomes, our soft calibration method with the data-adaptive tuning parameter selection can effectively reduce itself to hard calibration with little inflated mean squared errors.

\begin{figure}[!ht]
    \centering
    \includegraphics[width=.8\linewidth]{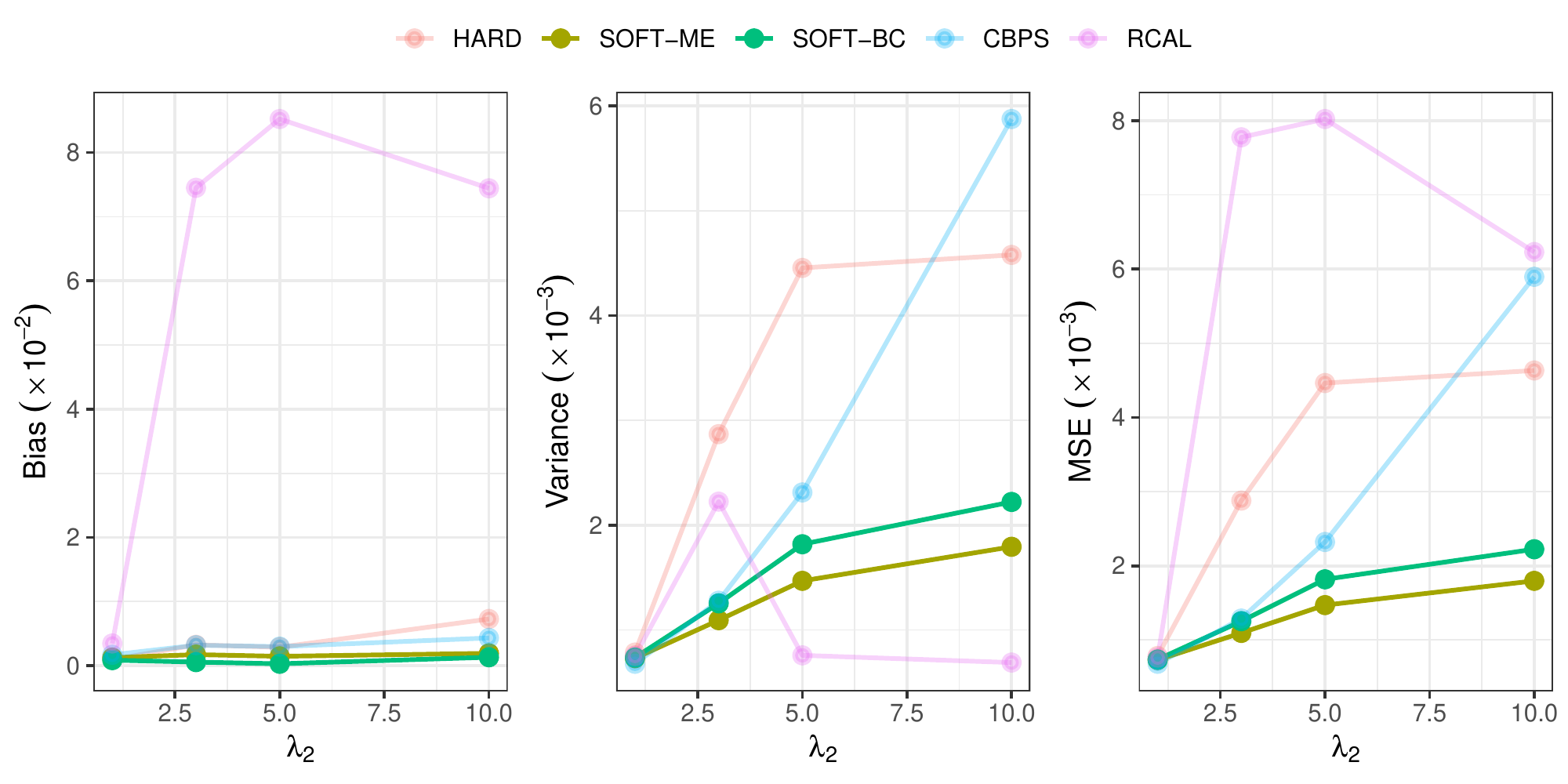}
    \includegraphics[width=.8\linewidth]{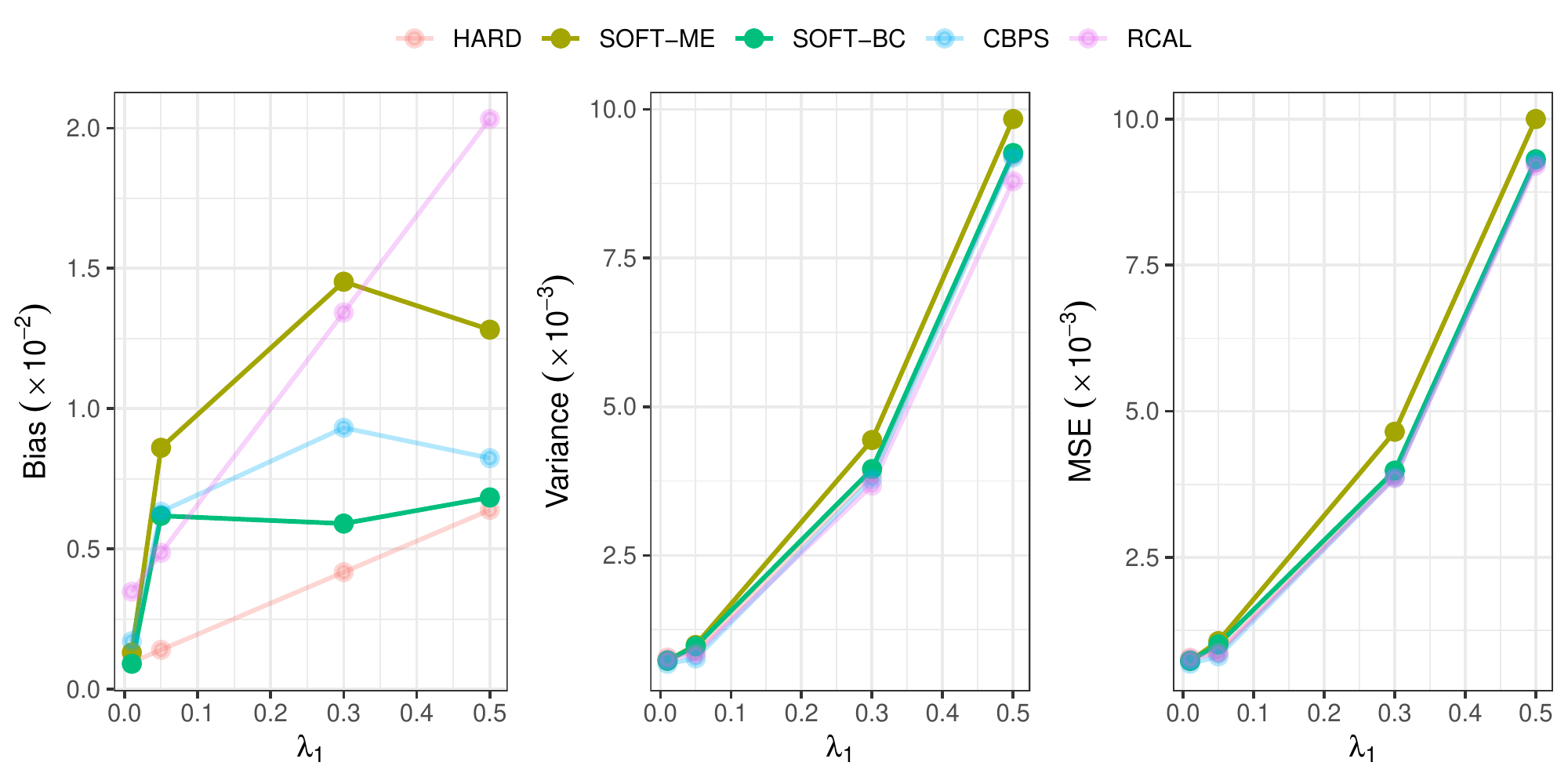}
    \caption{Bias $(\times 10^{-2})$, variance $(\times 10^{-3})$ and mean squared error $(\times 10^{-3})$ of the estimators with varying $\lambda_{2}$ with $\lambda_1=0.01$ (the top panel), and varying $\lambda_{1}$ with $\lambda_2 = 1$ (the bottom panel) under cluster setups $(k,n_i)=(30,200)$. In each plot, the estimators $\widehat{\theta}_{\rm hc}$, $\widehat{\theta}_w^{\me}$, $\widehat{\theta}_{\dr}$, $\widehat{\theta}_{\rm cbps}$ and $\widehat{\theta}_{\rm rcal}$ are labeled by \textsc{hard}, \textsc{soft-me}, \textsc{soft-bc}, \textsc{cbps} and \textsc{rcal}, where our proposed method is highlighted in bold.}
    \label{fig:varying}
\end{figure}

\subsection{Simulation Results for Cluster-specific Nonignorable Treatment Mechanism}\label{sec:sim_res_causal}

For an illustration of our soft calibration under the causal inference setting in $\S$\ref{sec:nonignore_trt}, we adopt the settings from \cite{yang2018propensity}. Two models are considered to generate the potential outcome for the $j$-th subject in cluster $i$ when $a=0,1$: 
\begin{align*}
&\textsc{lmm}: y_{ij}(a)=\bx_{ij}^{\T}\beta_{A=a}+\lambda_{1}a_i+e_{ij}\\
&\textsc{nlmm}: y_{ij}(a)=\bx_{ij}^{\T}\beta_{A=a}+
x_{1ij}^2 + x_{2ij}^2 + 0.1x_{3ij}^{\dagger} + 0.1x_{4ij}^{\dagger}+
\lambda_{1}a_i+e_{ij}
\end{align*}
where $\bx_{ij}=(1,x_{1ij},x_{2ij})^{\T}$, $x_{1ij}\sim U[-0.75,0.75]$, $x_{2ij}\sim N(0,1)$, ${\beta}_{A=0}=(0,1,1)^{\T}$, ${\beta}_{A=1}=(\tau,1,1)^{\T}$ with $\tau=2$, and $a_i\sim N(0,1)$, $e_{ij}\sim N(0,1)$
identically and independently. For the non-linear mixed-effects model, $x_{3ij}^{\dagger}$ and $x_{4ij}^{\dagger}$ are the standardized versions of $x_{3ij}=\exp(x_{1ij})$ and $x_{4ij}=\exp(x_{2ij})$, respectively. Also, we consider that treatment assignment $A_{ij}$ is generated through the logistic model:
$$P(A_{ij}=1\mid\bx_{ij},a_i)=h(-0.25+\bx_{ij}^{\T}{\alpha}+\lambda_{2}a_i)$$
with $h(\cdot)$ being the inverse logit-link and ${\alpha}=(0,1,1)^{\T}$. The observed outcome is $y_{ij}=A_{ij}y_{ij}(1)+(1-A_{ij})y_{ij}(0)$ with cluster setups $(k,n_i)=(30,200)$ and $(\lambda_1= 0.01, \lambda_2 = 1)$. Table \ref{tab:sim-causal:1} reports biases $(\times 10^{-2})$, variances $(\times 10^{-3})$, mean squared errors $(\times 10^{-3})$, and coverage
probability (\%) for causal inference with unmeasured cluster-level confounders. Notice that $\widehat{\theta}_w^{\sq}$ has the best performance under the linear mixed-effects model, but it does not have desirable cover probability when $y_{ij}$ follows a non-linear mixed-effects model as it postulates an erroneous linear model for the propensity score $P(A_{ij}=1\mid x_{ij},a_i)$. Other similar conclusions as stated in $\S$\ref{subsec:Cluster-specific-nonignorable-mi} can be also drawn from these results. 


\begin{table}[htbp]
\caption{\label{tab:sim-causal:1}Bias $(\times10^{-2})$, variance $(\times10^{-3})$, mean squared error
$(\times10^{-3})$ and coverage probability (\%) of the estimators under cluster-specific nonignorable treatment assignment based on $500$ simulated datasets when $(k,n_i)=(30,200)$}
\vspace{0.15cm}
\centering
{%
\begin{tabular}{lllllllllllll}
\toprule
&& 
$\widehat{\theta}_{\rm sim}$ &
$\widehat{\theta}_{\rm fix}$ & 
$\widehat{\theta}_{\rm rand}$ &
$\widehat{\theta}_{\rm hc}$ &
$\widehat{\theta}_{w}^{\sq}$ &
$\widehat{\theta}_{w}^{\me}$ &
$\widehat{\theta}_{\dr}$&
$\widehat{\theta}_{\rm cbps}$ & $\widehat{\theta}_{\rm rcal}$ \\ 
\midrule
\midrule
\multicolumn{10}{l}{\textit{Linear mixed-effects model with $(\lambda_1,\lambda_2)=(0.01,1)$}}\\
\midrule
Bias&&-165.1&0.00&-4.12&0.00&0.00&6.37&0.00&0.00&6.11\\
VAR&&4.23&3.28&2.94&1.44&1.01&1.23&1.26&1.17&1.34\\
MSE&&236.7&3.28&2.99&1.44&1.03&1.27&1.28&1.19&1.41\\
CP&&0.0&94.6&97.0&96.4&93.6&92.2&96.8&--&--\\
\\
\multicolumn{10}{l}{\textit{Linear mixed-effects model with $(\lambda_1,\lambda_2)=(0.01,10)$}}\\
\midrule
Bias&&-119.2&-50.4&-54.0&-2.31&3.78&3.40&3.16&2.10&110.8\\
VAR&&26.63&47.32&42.94&9.40&3.51&4.33&5.01&11.33&1.04\\
MSE&&780.0&292.1&293.3&9.45&3.67&4.52&5.18&11.46&25.73\\
CP&&0.0&28.4&26.4&93.6&94.6&92.8&94.2&--&--\\
\\
\multicolumn{10}{l}{\textit{Non-linear mixed-effects model with $(\lambda_1,\lambda_2)=(0.01,1)$}}\\
\midrule
Bias&&-60.75&-1.26&-1.33&0.00&16.91&3.31&3.31&11.55&3.36\\
VAR&&16.60&31.63&28.22&4.86&4.37&4.55&4.55&9.38&4.62\\
MSE&&141.4&31.70&28.32&4.87&6.82&4.57&4.58&11.76&4.78\\
CP&&20.2&93.0&94.0&93.4&87.0&92.2&95.8&--&--\\
\bottomrule
\end{tabular}}
\end{table}


\subsection{Additional Simulation Results}\label{sec:sim_res_cluster}

First, we present the densities of the calibration weights when $(k, n_i) = (30, 200)$ with $\lambda_1 = 0.01$ and $\lambda_2 = 10$ to handle cluster-specific nonignorable missingness (Figure \ref{fig:missing}, top) and cluster-specific nonignorable treatment assignment (Figure \ref{fig:missing}, bottom). It is obvious that $\widehat{\theta}_{\rm hc}$ calibration tends to reach extreme weights when the missing $\delta_{ij}$ or treatment mechanism $A_{ij}$ is extremely different across clusters, leading to an inefficient estimate. On the other hand, $\widehat{\theta}_{w}^{\me}$ calibration method can significantly mitigate the occurrence of extreme weights and therefore achieve more stable estimates.
\begin{figure}[!ht]
    \centering
    \includegraphics[width=.8\linewidth]{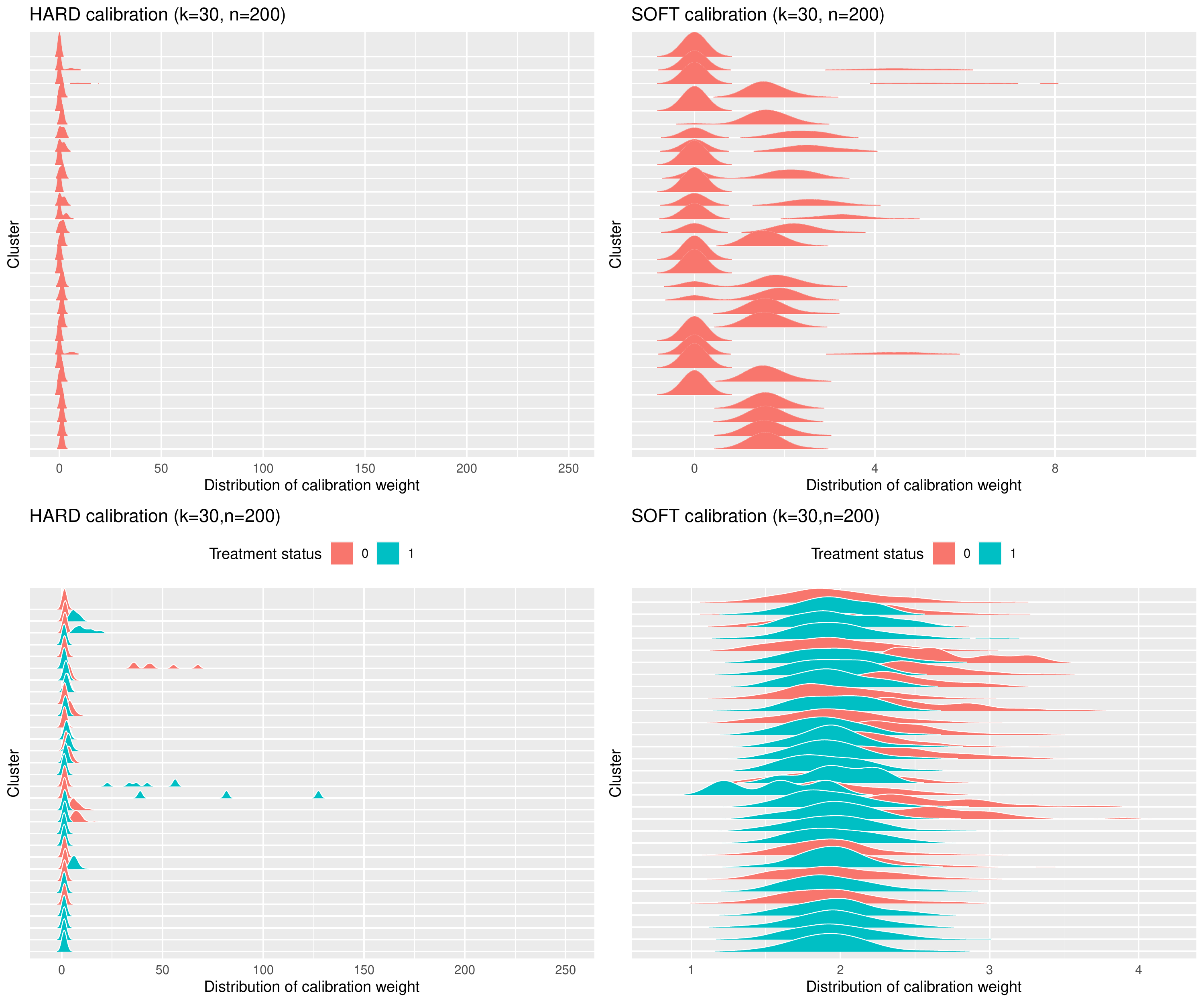}
    \caption{An illustration of the calibration weights using $\widehat{\theta}_{\rm hc}$ and $\widehat{\theta}_{w}^{\me}$ calibration method with $k=30, n_i= 200$ under cluster-specific nonignorable missingness (top) and cluster-specific nonignorable treatment assignment (bottom)}
    \label{fig:missing}
\end{figure}

Next, we present additional simulation results in Table \ref{tab:sim-other-cluster} with a different cluster setup $(k,n_i)=(100,30)$ to handle cluster-specific nonignorable missingness under the linear mixed-effects model illustrated in the main paper. One can observe that $\widehat{\theta}_{\rm rand}$ has lower variances  compared to $\widehat{\theta}_{\rm fix}$ by considering cluster effects as random, but it cannot control the bias well when the number of clusters $k$ is large. Moreover, when the between-cluster variation of $\delta_{ij}$ increases, i.e., $\lambda_2=10$, the exact calibration tends to yield extreme weights, whereas our soft-calibration estimators can mitigate the over-calibration issue and achieve more efficient estimation. In summary, the proposed estimators $\widehat{\theta}_{w}^{\sq}$, $\widehat{\theta}_{w}^{\me}$ and $\widehat{\theta}_{\dr}$ have smaller mean squared errors and exhibit desirable coverage probabilities in finite samples.

\begin{table}[!ht]
\caption{Bias $(\times10^{-2})$, variance $(\times10^{-3})$, mean squared error $(\times10^{-3})$ and coverage probability (\%) of the estimators under cluster-specific nonignorable missingness based on $500$ simulated datasets when $(k,n_i)=(100,30)$}
\vspace{0.15cm}
\centering
{%
\begin{tabular}{lllllllllllll}
\toprule
&& 
$\widehat{\theta}_{\rm sim}$ &
$\widehat{\theta}_{\rm fix}$ & 
$\widehat{\theta}_{\rm rand}$ &
$\widehat{\theta}_{\rm hc}$ &
$\widehat{\theta}_{w}^{\sq}$ &
$\widehat{\theta}_{w}^{\me}$ &
$\widehat{\theta}_{\dr}$&
$\widehat{\theta}_{\rm cbps}$ & $\widehat{\theta}_{\rm rcal}$ \\ 
\midrule
\midrule
\multicolumn{10}{l}{\textit{Linear mixed-effects model with $(\lambda_1,\lambda_2)=(0.01,1)$}}\\
\midrule
Bias&&21.22&0.44&1.49&0.03&0.22&0.19&0.13&0.12&1.50\\
VAR&&0.41&3.37&2.05&1.59&1.30&1.33&1.35&1.36&1.46\\
MSE&&45.42&3.39&2.27&1.59&1.30&1.33&1.35&1.36&1.68\\
CP&&0.0&96.0&94.0&94.2&94.2&92.2&92.0&--&--\\
\\
\multicolumn{10}{l}{\textit{Linear mixed-effects model with $(\lambda_1,\lambda_2)=(0.01,10)$}}\\
\midrule
Bias&&5.01&0.78&1.99&0.50&0.06&0.01&0.01&0.56&3.68\\
VAR&&0.45&5.44&3.69&3.02&1.78&1.85&2.01&5.07&1.28\\
MSE&&2.97&5.50&4.09&3.04&1.78&1.85&2.01&5.11&2.63\\
CP&&31.0&89.8&85.0&93.4&94.6&92.8&92.2&--&--\\
\bottomrule
\end{tabular}}
\label{tab:sim-other-cluster}
\end{table}

\section{Additional Application Results}\label{sec:additional_app}
Figure \ref{fig:A_Y} presents the frequency plot of $A_{ij}$ and kernel densities of $y_{ij}$ stratified by $A_{ij}$ for each fused cluster. The cluster-specific differences of $A_{ij}$ are not as evident as those of $y_{ij}$. Next, we fit the linear mixed effect models for the outcome $y_{ij}$ and use the maximum entropy loss function to obtain the weights $w$ for $A_{ij}=0$ and $A_{ij}=1$, respectively. Our cross-fitting strategy yields tuning parameters such as $\gamma_{n, A=1}=0.068$ and $\gamma_{n, A=0}=0.052$, which hardly relaxes the hard calibration for cluster effects. Therefore, it is reasonable to observe that the performance of the soft calibration is similar to the hard calibration in our application, which can also be justified by the visual similarities of the calibration weights produced by the hard and soft calibration schemes in Figure \ref{fig:weights}.

\begin{figure}[!ht]
    \centering
    \includegraphics[width = .8\linewidth]{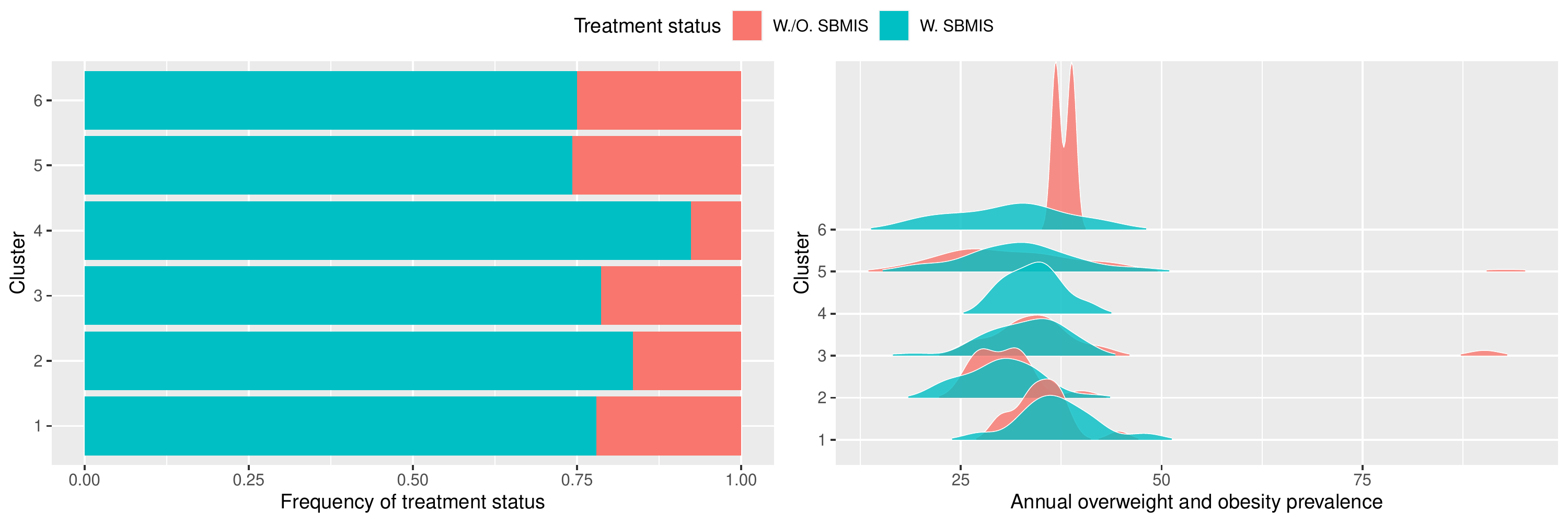}
    \caption{Frequency plot of the treatment $A_{ij}$ for each fused cluster (left); Kernel densities of the study variable $y_{ij}$ stratified by $A_{ij}$ for each fused cluster (right).}
    \label{fig:A_Y}
\end{figure}

\begin{figure}[!ht]
    \centering
    \includegraphics[width = .8\linewidth]{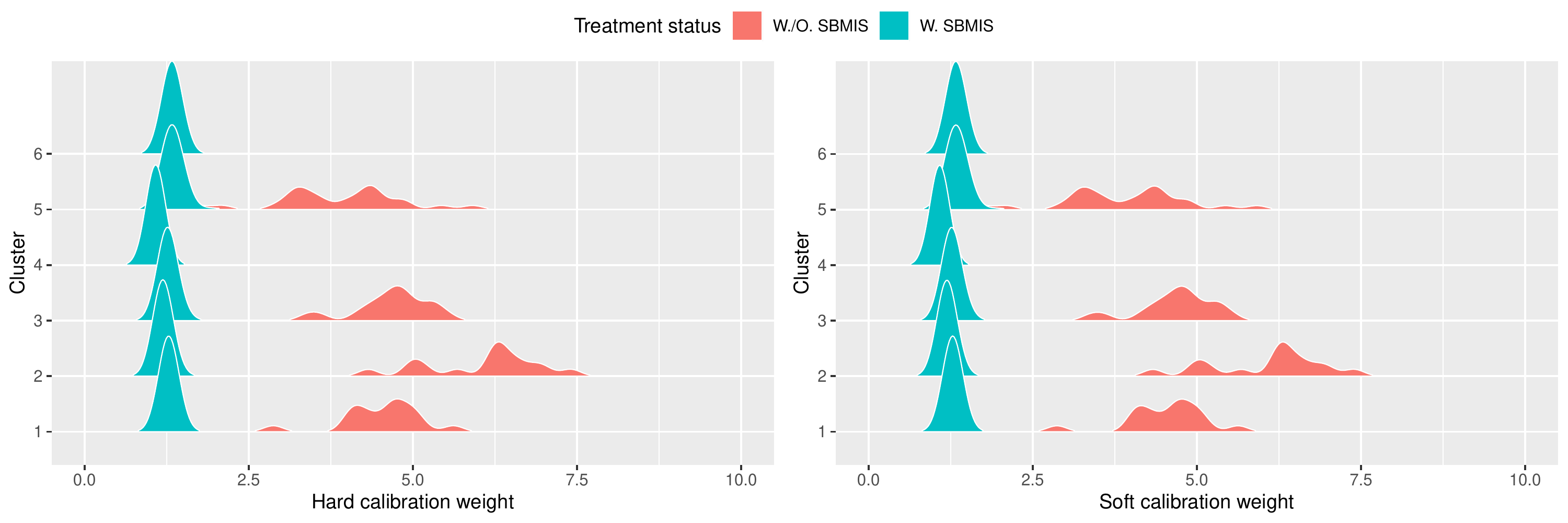}
    \caption{Kernel densities of the calibration weights $w_{ij}$ stratified by $A_{ij}$ for the hard calibration $\widehat{\theta}_{\rm hc}$ (left) and the soft calibration $\widehat{\theta}_w^{\me}$ (right).}
    \label{fig:weights}
\end{figure}

\section{Proof Additional Lemmas}

\subsection{Proof of Lemma \ref{lm:c_hat}}
\begin{proof}
Let $c^*$ be the solution to $\E\{U(c)\mid X,u\}=0$, we have by Taylor's Theorem, $$0=U(\widehat{c})=U(c^{*})+
\left\{\partial U(\overline{c})/\partial c^{\T}\right\}
(\widehat{c}-c^{*}),$$
where $\widehat{c},c^* \in \mathbb{C}_r^*$ defined in $\S$\ref{subsec:proof_unique}. For any $c\in\mathbb{C}_r^*$, we have
$$\partial U(c)/\partial c^{\T}	=\sum_{i\in\bbU}\delta_{i}w'(c^{\T}\bx_{i})\bx_{i}\bx_{i}^{\T} \leq \sup_{c\in\mathbb{C}_{r}^{*}}
\{N^{-1}nw'(c^{\T}\bx_{i})\} Nn^{-1} \sum_{i\in\bbU}\delta_{i}\bx_{i}\bx_{i}^{\T}  = O_\P(N),
$$ because $\mathbb{C}_r^*$ is compact and $w'$ is a continuous function. Next, we know that for any $c\in\mathbb{C}^*_r$, we have
\begin{align*}
   E\{\| U(c)\|^{2}\mid X_{\bbU},u\}&=\sum_{i\in\bbU}E[\{\delta_{i}w(c^{\T}\bx_{i})-1\}^{2}\bx_{i}^{\T}\bx_{i}\mid X_{\bbU},\bu]+
   \gamma_n^2\{\left(1_N^{\T} X_1 D_{12}+1_N^{\T} X_2 D_{22}\right) D_q^{-1}\}^{\otimes 2}\\
	&\leq\sum_{i\in\bbU}
 \E\left[\{\delta_{i}w(c^{\T}\bx_{i})^{2}-2\delta_{i}w(c^{\T}\bx_{i})+1\}\mid X_{\bbU}, u\right]\bx_{i}^{\T}\bx_{i}+O_{\P}(\gamma_n^2 n^{-2}N^2q)\\
	&\leq(Nn^{-1}\overline{d})^{2}n\Sigma_{n}+2Nn^{-1}\overline{d}n\Sigma_{n}+N\Sigma_{N}+O_{\P}(\gamma_n^2 n^{-2}N^2q),
\end{align*}
which is dominated by the first term under the condition that $\gamma_n=o(n^{1/2}q^{-1/2})$. Therefore,  $ U(c^{*})=O_{\P}(Nn^{-1/2})$ and $\widehat{c}-c^{*} =-\{\partial U(\overline{c})/\partial c^{\T}\}^{-1} U(c^{*})=O_{\P}(n^{-1/2})$.
\end{proof}

\subsection{Proof of Lemma \ref{lm:c_true}}
\begin{lemma}\label{lm:c_true}
Let the weight function or the inverse of propensity score function be correctly specified via $w(\cdot)$ entailed by the objective function $Q(w)$, that is, $w(c_t^{*\T}x_i)=P(\delta_i\mid x_i,u)^{-1}$, we have $\|\widehat{c}-c^*_t\|=O_{\P}(n^{-1/2})$ when $\gamma_n=o_{\P}(n^{1/2}q^{-1/2})$.
\end{lemma}
\begin{proof}
    The proof of Lemma \ref{lm:c_true} follows in a similar manner to the proof of Lemma \ref{lm:c_hat} as $N^{-1}\gamma_n(1_N^{\T}X_{1,\bbU}D_{12}+1_N^{\T}X_{2,\bbU}D_{22})D_q^{-1}=o(n^{-1/2})$ under the condition that $\gamma_n=O_{\P}(n^{1/2}q^{-1/2})$.
\end{proof}

\subsection{Proof of Lemma \ref{lem:unconstrained_convex}}
\begin{proof}
As it is known that
$$
g\{Q'(z)\} = Q'(z)z - 
Q\{(Q')^{-1}(z)\} = 
Q'(z)z - Q(z).
$$
By differentiating both sides with respect to $z$ once and twice, it yields
$$
g'\{Q'(z)\}Q''(z)=
Q''(z)z, \quad 
g''\{Q'(z)\}Q''(z) = 1.
$$
Further, it implies that 
$$
g''(z) = 
\frac{1}{Q''\{(Q')^{-1}(z)\}},
$$
which will be larger than zero due to the convexity of $Q(w)$. 
\end{proof}

\end{document}